\def\vMAFE{n} 
\if t\vMAFE\documentclass[smallextended] {svjour3}\else %
\documentclass[12pt,leqno]{article} %
\fi
\usepackage[english]{babel}
\if t\vMAFE %
\usepackage{graphicx}
\usepackage{amsmath, amsfonts, amssymb, bbm, mathrsfs,empheq}
\usepackage[dvipsnames]{xcolor}
\else
\usepackage[a4paper,%
tmargin=1.5cm,bmargin=1.5cm,lmargin=2cm,rmargin=2cm]{geometry}
\usepackage{amsmath, amsfonts, amsthm, amssymb, bbm, mathrsfs,empheq}

\usepackage[dvips]{graphicx}
\usepackage[dvipsnames]{xcolor}
\usepackage{hyperref}
\hypersetup{
    colorlinks,
    citecolor=OliveGreen,
    filecolor=blue,
    linkcolor=MidnightBlue,
    urlcolor=green
}
\fi
\title{Mean Field Game of Controls and %
\if t\vMAFE {}\else\\\fi %
An Application To Trade Crowding}
\if t\vMAFE %
\author{Pierre Cardaliaguet \and Charles-Albert Lehalle}
\institute{Pierre Cardaliaguet  \at 
  CEREMADE, Universit\'e Paris-Dauphine, Place du Marchal de Lattre de Tassigny, 75775 Paris cedex 16 (France)
  \and Charles-Albert Lehalle (corresponding author) \at
  Capital Fund Management (CFM), Paris (France) and Imperial College, London (UK) \email{c.lehalle@imperial.ac.uk} tel: +33 661 439 274.
}
\else
\author{Pierre Cardaliaguet\thanks{CEREMADE, Universit\'e Paris-Dauphine, Place du Marchal de Lattre de Tassigny, 75775 Paris cedex 16 (France)} {~}and Charles-Albert Lehalle\thanks{Capital Fund Management (CFM), Paris (France) and Imperial College, London (UK).}}
\fi

\newtheorem{Theorem}{Theorem}[section]

\newtheorem{Proposition}[Theorem]{Proposition}
\newtheorem{StylizedFact}{Stylized Fact}
\newtheorem{Lemma}[Theorem]{Lemma}

\newtheorem{Remark}[Theorem]{Remark}

\graphicspath{{images/}}
\DeclareGraphicsExtensions{.pdf,.png}

\if t\vMAFE
\def\answ#1{{\color{red}#1}}
\def\answca#1{{\color{red}#1}}
\else
\def\answ#1{#1}
\def\answca#1{#1}
\fi

\def\ds{\displaystyle}
\def\Esp{\mathbb{E}}
\def\R{\mathbb{R}}
\def\dive{{\rm div}}
\def\ep{\varepsilon}

\def\calF{{\cal F}}

\begin{document}
\maketitle

\begin{abstract}
  In this paper we formulate the now classical problem of optimal liquidation (or optimal trading) inside a Mean Field Game (MFG). This is a noticeable change since usually mathematical frameworks focus on one large trader \answ{facing} 
  a ``background noise'' (or ``mean field''). In standard frameworks, the interactions between the large trader and the price are a temporary and a permanent market impact terms, the latter influencing the public price.

In this paper the trader faces the uncertainty of fair price changes too but not only. He \answ{also} has to deal with price changes generated by other similar market participants, impacting the prices permanently too, and acting strategically.

Our MFG formulation of this problem belongs to the class of ``extended MFG'', we hence provide generic results to address these ``MFG of controls'', before solving the one generated by the cost function of optimal trading. We provide a closed form formula of its solution, and address the case of ``heterogenous preferences'' (when each participant has a different risk aversion). Last but not least we give conditions under which participants do not need to instantaneously know the state of the whole system, but can ``learn'' it day after day, observing others' behaviors.
{\if t\vMAFE %
 \fi}
\end{abstract}


%
\section{Introduction}

Optimal trading (or optimal liquidation) deals with the optimization of a trading path from a given position to zero in a given time interval.
Once a large asset manager \answ{takes}
 the decision to buy or to sell a large number of shares or contracts, he needs to implement his decision on trading platforms.
He \answ{has}
to be fast enough so that the traded price is as close as possible to his decision price, while he needs to take care of his market impact (
\answca{market impact is} the way his trading pressure moves the market prices, including his own prices, a detrimental way; see \cite[Chapter 3]{citeulike:12047995} for details).

The academic answers to this need \answ{goes}
from mean-variance frameworks
(initiated by \cite{OPTEXECAC00}) to more stochastic and liquidity
driven ones (see for instance \cite{citeulike:13300035}). Fine
modeling of the interactions between the price dynamics and the asset
manager trading process is difficult. The reality is probably a
superposition of a continuous time evolving process on the price
formation side and of an impulse control-driven strategy on the asset manager or trader side (see \cite{citeulike:5797837}).

The modeling of market dynamics for an optimal trading framework is sophisticated (\cite{citeulike:13250978} proposes a fine model of the orderbook bid and ask, \cite{ANYA05} \answ{suggests}
a martingale relaxation of the consumed liquidity, and \cite{citeulike:13125577} uses an Hawkes process, among others).
Part of the sophistication comes from the fact ``market dynamics'' are in reality the aggregation of the behaviour of other asset managers, buying or selling thanks to similar optimal schemes, behind the scene.
The usual framework for optimal trading is nevertheless the one of one
large privileged asset manager, facing a ``mean field'' (or a ``background noise'') made of the sum of the behaviour of other market participants.

An answer to the uncertainty on the components of an appropriate model of market dynamics (and especially of the ``market impact'', see empirical studies like \cite{citeulike:13497373} or \cite{citeulike:4325901} for details) could be to implement a robust control framework. Usually robust control consists in introducing an adversarial player in front of the optimal strategy, implementing systematically the worst decision (inside a well-defined domain) for the optimal player \answca{(see \cite{bernhard2003robust} for an extreme framework)}. Because of the adversarial player, one can hope \answ{that}  the uncertainty around the modeling cannot be used at his advantage by the optimal player.
To authors' knowledge, it has not been proposed for optimal trading by now.

Since we know the structure of market dynamics at our time scale of interest (i.e. a mix of players implementing similar strategies), another way to obtain robust result\answca{s} is to model directly this game, instead of the superposition of an approximate mean field and a synthetic sequences of adversarial decisions.


Our work takes place within the framework of Mean Field Games (MFG).  Mean Field Game theory studies optimal control problem\answca{s} with infinitely many interacting agents. The solution of the problem is an equilibrium configuration in which no agent has interest to deviate (a Nash equilibrium). The terminology and many substantial ideas were introduced in the seminal papers by Lasry and Lions \cite{citeulike:9313116,citeulike:9313152,citeulike:3614137}. Similar models were discussed at the same time by Caines, Huang and Malham\'e (see, e.g., \cite{huang2006large}), who computed explicit solutions for the linear-quadratic (LQ) case (see also \cite{bensoussan2016linear}): we use these techniques in our specific framework. Applications to economics and finance were first developed in   \cite{gueant2011mean}. Since these pioneering works the literature on MFG has grown very fast: see for instance the monographs or the survey papers \cite{bensoussan2013mean,caines2015mean,gomes2014mean}.

The MFG system considered here for the optimal trading differs in a substantial way from the standard ones by the fact that the mean field is not, as usual, on the position of the agents, but on their controls: this specificity is dictated by the model. 
Similar---more general---MFG systems were introduced in \cite{gomes2014existence} under the terminology of {\it extended} MFG models.
 In \cite{gomes2014existence,gomes2016extended}, existence of solutions is proved for deterministic MFG under suitable structure assumptions. A chapter in the monograph \cite{carmonadelaruebook}  is also devoted to this class of MFG, with a probabilistic  point of view. Since our problem naturally involves a degenerate diffusion, we provide a new and general existence result of a solution in this framework. Our approach is related with techniques for standard MFG of first order with smoothing coupling functions (see \cite{citeulike:3614137} or \cite{cardaliaguet2015learning}). 

As the MFG system is an equilibrium configuration---in which theoretically each agent has to know how the other agents are going to play in order to act optimally---it is important to explain how such a configuration can pop up in practice. This issue is called ``learning" in game theory and has been the aim of a huge literature (see, for instance, the monograph \cite{fudenberg1998theory}). The key idea is that an equilibrium configuration appears---without coordination of the agents---because the game has been played sufficiently many times. {
 In MFG theory, the closely related concept of adaptative control was implemented, for infinite horizon problems, in \cite{KizilkaleCaines,nourian2012mean}. The first explicit reference} to learning in MFG theory can be found in \cite{cardaliaguet2015learning}. This idea seems very meaningful for our optimal trading problem where the asset managers are lead to buy or sell daily a large number of shares or contracts. We implement the learning procedure within the simple LQ-framework: we show that, when the game has been repeated a sufficiently large number of times, the agents---without coordination and subject to measurement errors---actually implement trading strategies which are close to the one suggested by the MFG theory. 


The starting point of this paper is an \emph{optimal liquidation problem}: a trader has to buy or sell (we formulate the problem for a sell) a large amount of shares or contracts in a given interval of time (from $t=0$ to $t=T$). The reader can think about $T$ as typically being a day or a week. Like in most optimal liquidation (or optimal trading) problems, the utility function of the trader \answ{has}
three components: the state of the cash account at terminal time $T$ (i.e. the more money the trader obtained from the large sell, the better), a liquidation value for the remaining inventory (with a penalization corresponding to a large market impact for this large ``block''), and a risk aversion term (corresponding to not trading instantaneously at his decision price). As usual again, the price dynamics will be influenced by the sells of the trader (via a \emph{permanent market impact} term: the faster \answ{the trader} trades, the more he impacts the price \answ{in} a detrimental way); on the top of this cost, the trader will suffer from \emph{temporary market impact}: this will not change the public price but his own price (the reader can think about the ``cost of liquidity'', like a bid-ask spread cost). In their seminal paper \cite{OPTEXECAC00}, Almgren and Chriss noticed such a framework gives birth to an interesting optimization problem for the trader: on the one hand if he trades too fast he will suffer from market impact and liquidity costs on his price, but on the other hand if he trades too slow, he will suffer from a large risk penalization (the ``fair price'' will have time to change a detrimental way). Once expressed in a dynamic (i.e. time dependent) way, this optimization turns into a stochastic control problem.
For details about this standard framework, see \cite{cartea15book}, \cite{olivier16book} or \cite[Chapter 3]{citeulike:12047995}.

In the standard framework, the  trader faces a \emph{mean field}: a Brownian motion he is the only one to influence via his permanent market impact. \answ{This will no longer be the case in this paper:}
on the top of a Brownian motion (corresponding to the unexpected variations of the \emph{fair price} while market participants are trading), we  add the consequences of the actions of a \emph{continuum of market participants}. Each \answ{participant has}
to buy or sell a number of shares or contracts $q$ (positive for sellers and negative for buyers).
 This continuum is characterized by the density of the remaining inventory of participants $dm(t,q)$. A variable of paramount importance is the net inventory of all participants: $E(t):=\int_q q \, m(t,q)\, dq$. 

To Authors' knowledge, only two papers are related to \answ{our approach:} 
one by Carmona et al., using MFG for fire sales \cite{citeulike:13641571}, and one by Jaimungal and Nourian \cite{citeulike:13586166} for optimal liquidation of one large trader in front of smaller ones. 
They are nevertheless different \answ{from}
ours: the first one (fire sales) has not the same cost function \answ{as in}
optimal liquidation \answ{problems}, and the second one \answ{investigates the behavior of}
a large trader (having to sell substantially more shares or contracts than the others, with a risk aversion) 
facing \emph{a crowd of small traders} (with a lot of small inventories and no risk aversion). The topic of this second paper is more the one of a large asset manager trading in front of high frequency traders.

\answ{In} this paper \answ{we assume that} the public price \answ{is}
influenced by the permanent market impact of all market participants.
\answ{Note that, conversely, all market participants face the public price thus affected.}
 It corresponds to the day-to-day reality of traders: it is not a good news to buy while other participants are buying, but it is good to have to buy while others are selling. 
\answ{In such a configuration, the} participants 
act strategically, taking into account all the information they have. \answ{As explained in Section \ref{sec:optrade}, this leads to a Nash equilibrium of MFG type, in which the mean field depends on the agents' actions. This Nash MFG equilibrium can be summarized by a system  of forward-backward PDEs (see \eqref{eq:master}), coupling a (backward) Hamilton-Jacobi equation with a (forward) Kolmogorov equation. When the preferences of the agents are  homogeneous, this system can be solved explicitly and displays interesting---and not completely intuitive---features  (Section \ref{sec:identical}). For instance, one can notice that the coefficient affecting the permanent market impact has a strong influence on the whole system: the highest the coefficient, the fastest the market players have to drive their inventory to zero. Another interesting situation is when a participant is close to a zero inventory (or for low terminal constraints and low risk aversion). It can then be rewarding to ``follow the crowd'': a seller may have interest to buy for a while. These qualitative conclusions are summarized }
as ``Stylized Facts'' at the end of Section \ref{sec:identical}. 

\answ{The reader might object that an equilibrium configuration as described in Section  \ref{sec:optrade} or \ref{sec:identical}
is unlikely to be observed  because, in practice, the market participants compute their optimal strategy for  an optimal control problem, and not for a game. We address this issue in Section \ref{sec:equilearn}, in  the case where each participant has his own risk aversion (i.e. a case of \emph{heterogenous preferences}, in game theory terminology). We first discuss the existence and the uniqueness of a MFG Nash equilibrium in this more general framework. Then we model---in a slightly more realistic way---the day-after-day behavior of the market participants: we explain that, as 
market participants observe a (possibly) noisy measurement of the daily net trading speed of all investors,  they can try to derive from their past observations an approximation of the permanent market impact for the next day and compute their optimal strategy accordingly.  We show that, doing so, \emph{they end up playing  a Nash MFG equilibrium}. Let us underline that, in our model, the market participants do not have  access to the distribution of the trading positions of the other participants; they do not necessarily have the same estimate of the permanent market impact; they are not even  aware that they are  ``playing a game". Nevertheless, the configuration after stabilization is an MFG equilibrium.} 

Since we had to develop our own MFG tools to address this case (the mean field involving the \emph{controls of the agents}, not in their state), the last part of this paper (Section \ref{sec:generic})  addresses  this kind of MFG systems in a very generic way. The reader will hence find in this Section tools to handle such problems, with generic cost functions and not only with the ones of the usual optimal liquidation problem. \answ{Note that we only address here the question of well-posedness of the game with infinitely many agents: the application to games with a finite number of players, as well as the learning procedures will be developed in future works.}
\\

{\bf Acknowledgement:} Authors thank Marc Abeille for his careful reading of Section \ref{sec:redone}. The first author was partially supported by the ANR (Agence Nationale de la Recherche) projects ANR-14-ACHN-0030-01
and ANR-16-CE40-0015-01.

\section{Trading Optimally Within The Crowd}
\label{sec:optrade}

\subsection{Modeling a Mean Field of Optimal Liquidations}\label{subsec:model}

A continuum of investors indexed by $a$ decide to buy or sell a given tradable instrument.
The decision is a signed quantity $Q^a_0$ to buy (in such a case $Q^a_0$ is negative: the investor has a negative inventory at the initial time $t=0$) or to sell (when $Q^a_0$ is positive).
All investors have to buy or sell before a given terminal time $T$, each of them will nevertheless potentially trade faster or slower since each of them will be submitted to 
different risk aversion parameters $\phi^a$ and $A^a$. The distribution of the risk aversion parameters is independent to anything else.

Each investor will control its trading speed $\nu^a_t$ through time, in order to fullfil its goal.
The price $S_t$ of the tradable instrument is submitted to two kinds of move\answca{s}: an exogenous innovation supported by a standard Wiener process $W_t$ (with its natural probability space and the associated filtration $\calF_t$), and the \emph{permanent market impact} generated linearly from the buying or selling pressure $\alpha \mu_t$ where $\mu_t$ is the net sum of the trading speed of all investors (like in \cite{citeulike:13587586}, but in our case $\mu_t$ is endogenous where it is exogenous in their case) and $\alpha>0$ is a fixed parameter.
\begin{equation}
  \label{eq:price}
  dS_t = \alpha \mu_t\, dt + \sigma\, dW_t.
\end{equation}

The state of each investor is described by two variables: its inventory $Q^a_t$ and its wealth $X^a_t$ (starting with $X^a_0=0$ for all investors). The evolution of $Q^a$ reads
\begin{equation}
  \label{eq:qty}
  dQ^a_t = \nu^a_t\, dt,
\end{equation}
since for a seller, $Q^a_0>0$ (the associated control $\nu^a$ will be mostly negative)
and the wealth suffers from linear trading costs (or \emph{temporary}, or \emph{immediate  market impact}\answca{, parametrized by $\kappa$}):
\begin{equation}
  \label{eq:wealth}
  dX^a_t = -\nu^a_t (S_t + \kappa \cdot\nu^a_t)\, dt.
\end{equation}
Meaning the wealth of a seller will be positive (and the faster you sell --i.e. $\nu^a$ is largely negative--, the smaller the sell price).

The cost function of investor $a$ is similar to the ones used in \cite{cartea15book}: it is made of the wealth at $T$, plus the value of the inventory penalized by a terminal market impact, and minus a running cost quadratic in the inventory:
\begin{equation}
  \label{eq:Value}
  V^{a}_t:=\sup_\nu \Esp\left( X^a_T + Q^a_T (S_T - A^a\cdot Q^a_T) - \phi^a \int_{t}^T (Q^a_s)^2\, ds \right).
\end{equation}
\answca{We use this cost function by purpose: a lot of efforts have been made around this utility function by Cartea, Jaimungal and their different co-authors to show it can emulate (provided some changes) most of costs functions provided by brokers to dealing desks of asset managers. This specific one emulates an \emph{Implementation Shortfall} algorithms, while it can be changed to emulate a \emph{Percentage of Volume} or a \emph{Volume Weighted Average Price} (see \cite[Capter 3]{citeulike:12047995} for a list of common trading algorithms). Such changes in the cost function would nevertheless impact the paper and demand some complementary work.}

%
\subsection{The Mean Field Game system of Controls}

The Hamilton-Jacobi-Bellman associated to (\ref{eq:Value}) is
\begin{equation}
  \label{eq:HJB}
  0 = \partial_t V^a - \phi^a\, q^2 + \frac{1}{2}\sigma^2\partial^2_S V^a +\alpha \mu \partial_S V^a+  \sup_\nu \left\{   \nu \partial_q V^a - \nu (s+\kappa\, \nu)\partial_X V^a \right\}
\end{equation}
with terminal condition 
$$
V^a(T,x,s,q;\mu)= x+q(s-A^aq).
$$ 

Following the Cartea and Jaimungal's approach, we will use the following ersatz: 
\begin{equation}
  \label{eq:ersatz}
  V^a=x+q s + v^a(t,q;\mu).
\end{equation}
Thus the HJB on $v$ is
$$-\alpha\mu\, q = \partial_t v^a - \phi^a\, q^2 + \sup_\nu \left\{ \nu \partial_q v^a -\kappa\, \nu^2 \right\}$$
with terminal condition
$$
v^a(T,q;\mu)= -A^aq^2
$$
and the associated optimal feedback is
\begin{equation}
  \label{eq:optnu}
  \nu^a(t,q) = \frac{\partial_q v^a(t,q)}{2\kappa}.
\end{equation}

\paragraph{Defining the mean field.}
The expression of the optimal control (i.e. trading speed) of each investor shows that the important parameter for each investor is its current inventory $Q^a_t$. The mean field of this framework is hence the distribution $m(t,dq,da)$ of the inventories of investors and of their preferences. Its initial distribution is fully specified by the initial targets of investors and the distribution of the $(\phi^a, A^a)$.

It is then straightforward to write the net trading flow $\mu$ at any time $t$
\begin{equation}
  \label{eq:mudyn}
  \mu_t = \int_{(q,a)} \nu^a_t(q)\, m(t,dq,da) = \int_{q,a} \frac{\partial_q v^a(t,q)}{2\kappa} \, m(t,dq,da).
\end{equation}
Note that implicitly $v^a$ is a function of $\mu$, meaning we will have a fixed point problem to solve in $\mu$.

We now write the 
evolution of the density $m(t,dq,da)$ of $(Q_t^a)$. By the dynamics (\ref{eq:qty}) of $Q_t^a$, we have 
$$
\partial_t m+\partial_q \left( m \frac{\partial_q v^a}{2\kappa}\right)=0
$$
with initial condition $m_0=m_0(dq,da)$ (recall that the preference $(\phi^a, A^a)$ of an agent $a$ is fixed during the period). 

\paragraph{The full system.} Collecting the above equations we find our twofolds mean field game system made of the backward PDE on $v$ coupled with the forward transport equation of $m$: 
\begin{equation}
  \label{eq:master}
\left\{  \begin{array}{rcl}
\displaystyle -\alpha  q\,\mu_t &=& \displaystyle\partial_t v^a - \phi^a\, q^2 + \frac{(\partial_q v^a)^2}{4\kappa}\\
0 &= & \partial_t m+\partial_q \left( m \frac{\partial_q v^a}{2\kappa}\right)\\
\mu_t & = &\ds  \int_{(q,a)} \frac{\partial_q v^a(t,q)}{2\kappa} \, m(t,dq,da) 
  \end{array} \right.
\end{equation}
The system is complemented with the initial condition (for $m$) and terminal condition (for $v$): 
$$
m(0,dq,da)=m_0(dq,da), \qquad v^a(T,q;\mu)= -A^aq^2.
$$

\answ{The above system is interpreted as a Nash equilibrium configuration in a game with infinitely many market participants: a (small) market participant, anticipating the net trading flow $(\mu_t)$, computes his optimal strategy by solving an optimal control problem which, after simplification, leads to the  equation for $v^a$ coupled with the terminal condition. When all market participants trade optimally, the distribution $m$ of the inventories and preferences evolves according to the second equation, complemented with the initial for $m$. Then one derives the net trading flow $(\mu_t)$ 
\answca{as a} function of $m$ and $v^a$ through the third equation.  }

\section{Trade crowding with identical preferences}\label{sec:identical}

In this section, we suppose that all agents have identical preferences: $\phi^a=\phi$ and $A^a=A$ for all $a$. The main advantage of this assumption is that it leads to explicit solutions. 

\subsection{The system in the case of identical preferences} To simplify notation, we omit the parameter 
$a$ in all expressions. \answ{We aim at solving \eqref{eq:master} (in which $a$ is omitted). }
It is convenient to set $\displaystyle E(t)= \Esp\left[Q_t\right]= \int_q qm(t,dq)$. Note that 
$$
E'(t)= \int_q q\partial_t m(t,dq),
$$
so that, using the equation for $m$ and an integration by parts: 
\begin{equation}
  \label{eq:mdyn}
  E'(t)=  -\int_q q\partial_q \left( m(t,q) \frac{\partial_q v(t,q)}{2\kappa}\right)dq= \int_q  \frac{\partial_q v(t,q)}{2\kappa}m(t,dq).
\end{equation}

\subsection{Quadratic Value Functions}

When $v(t,q)$ can be expressed as a quadratic function of $q$:
$$v(t,q) = h_0(t) + q\, h_1(t) - q^2 \, \frac{h_2(t)}{2},$$
then the backward part of the master equation can be split in three parts
\begin{equation}
 \label{eq:lqe1}
 \begin{array}{rcl}
  0 &=& \displaystyle - 2\kappa h_2'(t) - 4\kappa \phi+ (h_2(t))^2,\\
  \displaystyle 2\kappa\alpha\mu(t) &=&  \displaystyle  - 2\kappa h_1'(t) + h_1(t)h_2(t),\\
  \displaystyle -{(h_1(t))^2} &=& \displaystyle {4\kappa} h_0'(t) ,
\end{array}
\end{equation}
One also has to add the terminal condition: as $V_T= x+q(s-Aq)$, $v(T,q)= -Aq^2$. This implies that 
\begin{equation}\label{eq:lqe1IC}
h_0(T)=h_1(T)=0, \; h_2(T)= 2A.
\end{equation}

\paragraph{Dynamics of the mean field.} Recalling \eqref{eq:mudyn}, we have 
\begin{equation}\label{eq:mumu}
  \mu(t) =  \int_q \frac{\partial_q v(t,q)}{2\kappa} \, dm(q)= \int_q \frac{h_1(t)-qh_2(t)}{2\kappa} \, dm(q)=  \frac{h_1(t)}{2\kappa} - \frac{h_2(t)}{2\kappa}E(t).
  \end{equation}
Moreover, by (\ref{eq:mdyn}), we also have 
$$
 E'(t)= \int_q m(t,q) \left( \frac{h_1(t)}{2\kappa} - \frac{h_2(t)}{2\kappa}q \right)dq = \frac{h_1(t)}{2\kappa} -  \frac{h_2(t)}{2\kappa}E(t).
$$
So we can supplement \eqref{eq:lqe1} with 
\begin{equation}\label{eq:lqe2}
2\kappa E'(t) =  h_1(t)  - E(t) \cdot h_2(t).
\end{equation}

\paragraph{Summary of the system.} We now collect all the equations. Recalling \eqref{eq:mumu}, 
we find: 
\begin{subequations}
\begin{empheq}[left=\left\{,right=\right.]{align}
    \label{eq:dh2} \displaystyle 4\kappa \phi &= -2\kappa h_2'(t)+ (h_2(t))^2,\\ 
    \label{eq:dh1} \rule{3ex}{0px}\displaystyle \alpha h_2(t)  E(t)&=  2\kappa h_1'(t)+ h_1(t) \left(\alpha - h_2(t)\right),\\
    \label{eq:dh0} \displaystyle -{(h_1(t))^2} &= {4\kappa}\displaystyle  h_0' (t),\\
    \label{eq:dEq} \displaystyle    2\kappa E'(t) &=  h_1(t)  - h_2(t)E(t). 
  \end{empheq}
\end{subequations}
with the boundary conditions 
$$
h_0(T)=h_1(T)=0, \; h_2(T)= 2A, \; E(0)=E_0,
$$
where $E_0=\int_q qm_0(q)dq$ is the net initial inventory of market participants (i.e. the expectation of the initial density $m$).

\subsubsection{Reduction to a single equation} 
\label{sec:redone}

From now on we consider $h_2$ as given and derive an equation satisfied by $E$. By \eqref{eq:dEq}, we have 
$$
h_1= 2\kappa  E' + h_2E,
$$
so that 
\begin{equation}
  \label{eq:hEh}
   h_1'= 2\kappa E'' + E h_2'+ h_2 E'.
\end{equation}

Plugging these expressions into \eqref{eq:dh1}, we obtain 
$$
\begin{array}{rl}
\ds 0 \; = & \ds   h_1'+h_1 \frac{\alpha - h_2}{2\kappa} - \frac{\alpha h_2}{2\kappa} E, \\
= & \ds  2\kappa E'' + E h_2'+ h_2 E'+\left(2\kappa E' + h_2E\right)  \frac{\alpha - h_2}{2\kappa}- \alpha\frac{h_2}{2\kappa} E\\
= &\ds 2\kappa E'' + \alpha E'  + 2 E \left(\frac{1}{2} h_2' -\frac{(h_2)^2}{4\kappa} \right).
\end{array}
$$
Recalling \eqref{eq:dh2} we find
$$
0= 2\kappa E'' + \alpha E'  -2\phi E.
$$
The boundary conditions are $E(0)=E_0$, $h_2(T)=2A$, $h_1(T)=0$, where the last expression can be rewritten by taking \eqref{eq:dEq} into account.
To summarize, the equation satisfied by $E$ is: 
\begin{equation}\label{eq:systE}
\left\{\begin{array}{l}
\displaystyle 0=2\kappa E'' (t)+ \alpha E' (t) -2\phi E(t)\qquad {\rm for}\; t\in (0,T),\\
\displaystyle E(0)=E_0, \qquad \kappa E'(T) + AE(T)=0.
\end{array}\right.
\end{equation}

\subsubsection{Solving \eqref{eq:systE}}

After some easy but tedious computation explained in Appendix \ref{sec:app:prop:E}, one finds: 

\begin{Proposition}[Closed form for the net inventory dynamics $E(t)$]
\label{prop:E}
For any $\alpha\in \R$, the problem \eqref{eq:systE} has a unique solution $E$,  given by 
$$
E(t)= E_0 a \left( \exp\{r_+t\}-\exp\{r_-t\}\right)+ E_0\exp\{r_-t\}
$$
where $a$ is given by
$$
a= \frac{(\alpha/4+\kappa\theta-A) \exp\{-\theta T\} }{-\frac{\alpha}{2}{\rm sh}\{\theta T\} +2\kappa\theta{\rm ch}\{\theta T\}+2A {\rm sh}\{\theta T\}},
$$
the denominator being positive and the constants $r^\pm_\alpha$ and $\theta$ being given by 
$$
r_\pm := -\frac{\alpha}{4\kappa} \pm \theta, \qquad \theta:=\frac{1}{\kappa} \sqrt{\kappa\phi+\frac{\alpha^2}{16}}.
$$ 
Moreover, 
\begin{equation}
  \label{eq:h2:gen}
  h_2(t) = 2\sqrt{\kappa\phi}\frac{1+c_2e^{rt}}{1-c_2e^{rt}},
\end{equation}
where $r=2\sqrt{\phi/\kappa}$ and 
$$c_2=\answca{-}\frac{1 - A/\sqrt{\kappa \phi}}{1 + A/\sqrt{\kappa\phi}} \cdot e^{-rT}.$$
\end{Proposition}

\begin{Remark}{\rm The last needed component to obtain the optimal control using \eqref{eq:optnu} is $h_1(t)$. Thanks to \eqref{eq:hEh}, it can be easily written from $E$, $E'$ and $h_2$ (note $h_2$ is mostly negative for our sell order): 
$$h_1(t) = 2\kappa\cdot E'(t) + h_2(t)\cdot E(t).$$

This gives an explicit formula for the optimal control for any value of the parameters:
$\kappa$ (the instantaneous market impact), $\phi$ (the risk aversion), $\alpha$ (the permanent market impact), $A$ (the terminal penalization), $E_0$ (the initial net position of all participants), and $T$ (the duration of the ``game'').
}
\end{Remark}

\begin{figure}[!ht]
  \centering
  \includegraphics[width=\textwidth]{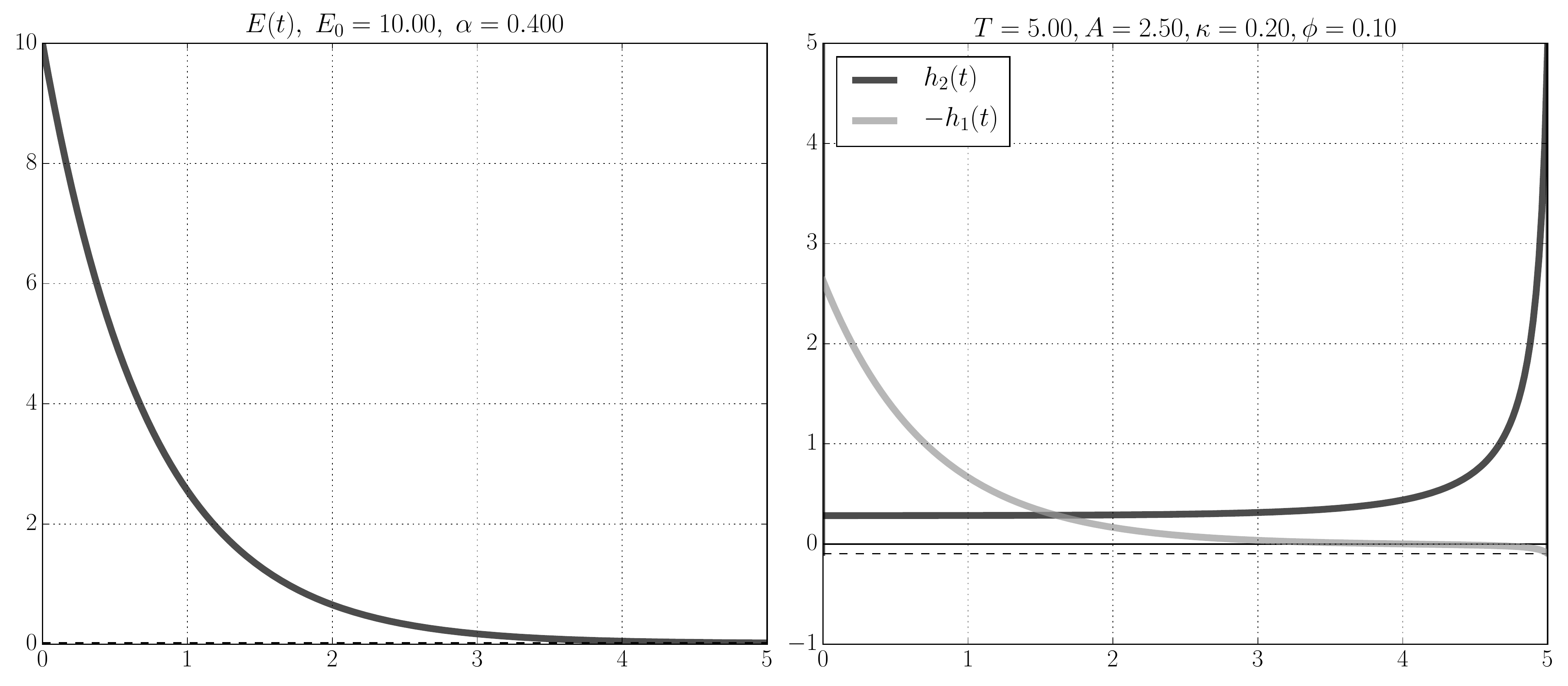}
  \caption{Dynamics of $E$ (left) and $-h_1$ and $h_2$ (right) for a standard set of parameters: %
    $\alpha=0.4$, $\kappa=0.2$, $\phi=0.1$, $A=2.5$, $T=5$, $E_0=10$.}
  \label{fig:stdhEh}
\end{figure}

\paragraph{Typical Dynamics.}
Figure \ref{fig:stdhEh} shows the joint dynamics of $E$ (left panel), $-h_1$ and $h_2$ (right panel) for typical values of the parameters: $\alpha=0.4$, $\kappa=0.2$, $\phi=0.1$, $A=2.5$, $T=5$, $E_0=10$.
As expected, $E(t)$ goes very close to 0 at $t=T$; in our reference case $E(T)=0.02$. 
  Looking carefully at Proposition \ref{prop:E}, it can be seen the main component driving $E(T)$ to zero is $\exp r^\alpha_- T$, where $r^\alpha_-=-\alpha / 4 - (\sqrt{\kappa\phi + \alpha^2/16})/\kappa$, hence:
\begin{itemize}
\item The best way to obtain a terminal net inventory of zero is to have a large $\alpha$, or a large $\phi$, or a small $\kappa$. Surprisingly, having a large $A$ does not help that much. It mainly urges the trading very close to $t=T$ when the other parameters decrease $E$ earlier.
\item $h_2$ increases slowly to $2A$, while $h_1$ goes from a negative value to a slightly positive one.
\end{itemize}

To understand the respective roles of $h_2$ and $h_1$, one should keep in mind the optimal control is $(h_1(t) - q h_2(t))/(2\kappa)$. Having a negative $h_1$ increases the trading speed \answca{of a seller}. That's why we draw $-h_1$ instead of $h_1$ on all the figures.

\answca{Since participants influences themselves via $\alpha\mu_t$ the permanent market impact coefficient times the sum of their controls (that are functions of the mean field), one can consider the lower $\alpha$, the more ``disconnected'' players from the influence of the mean field.}

\answca{\begin{StylizedFact}[Influence of the mean field varies with $\alpha$]\label{prop:smallalpha}
Figure \ref{fig:E:small:h1} compares the components of the optimal strategies for two values of $\alpha$ (the strength of the influence of the players one on each others): 
when players are less connected: $h_2(t)$ does not change and $h_1(t)$ is smaller, except at the end of the trading.
\end{StylizedFact}}

\answca{Keep in mind the optimal control $\nu_t(q)$ is proportional to $h_1(t)-q\,h_2(t)$. This means when $t$ is close to zero (i.e. start of the trading), $q$ is close to $Q_0$, and hence $q\,h_2(t)$ is large compared to $h_1(t)$.}

\begin{figure}[!ht]
  \centering
  \includegraphics[width=\textwidth]{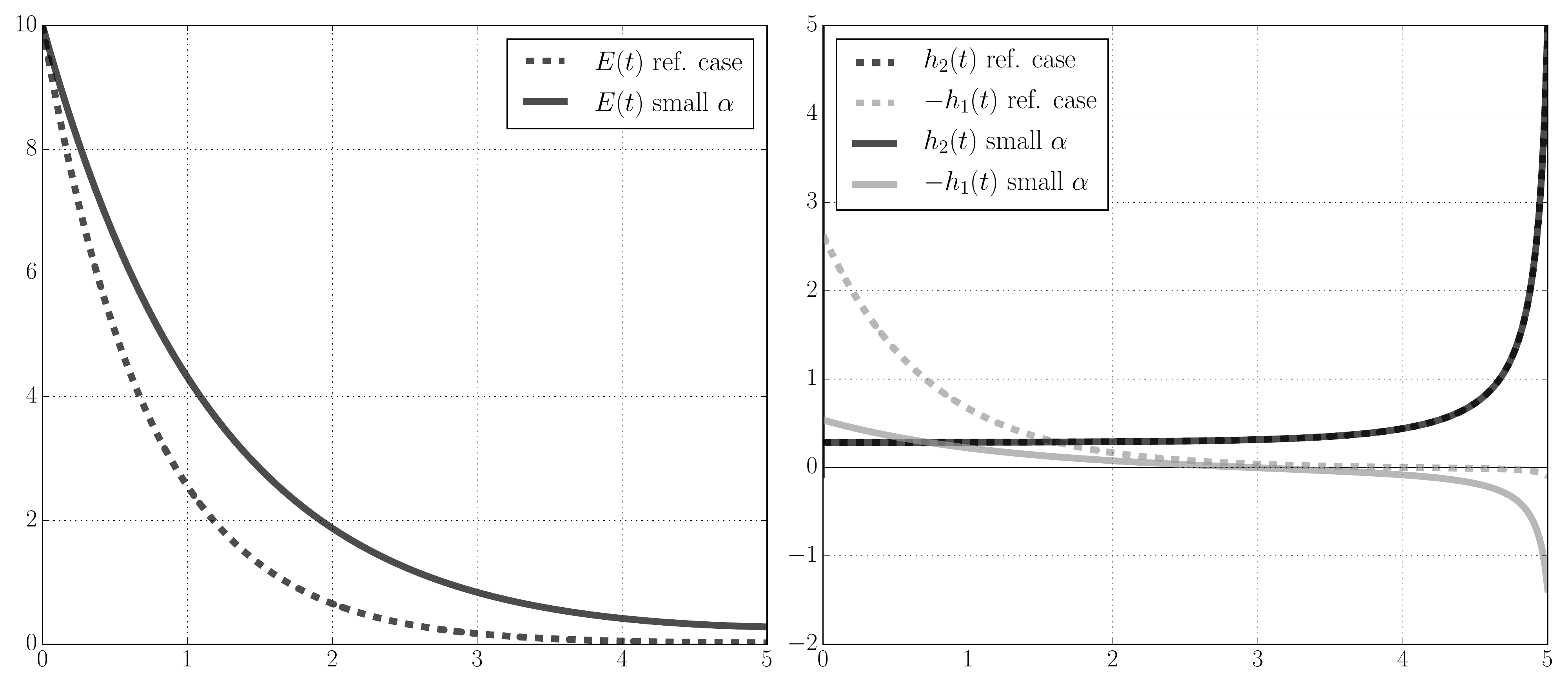}
  \caption{Comparison of the dynamics of $E$ (left) and $-h_1$ and $h_2$ (right) between the ``reference'' parameters of Figure \ref{fig:stdhEh} and smaller $\alpha$ (i.e. $\alpha=0.1$ instead of $0.4$) such that $|h_1(0)|$ is smaller.}
  \label{fig:E:small:h1}
\end{figure}

\begin{StylizedFact}[Driving $E$ to zero]\label{prop:Edrive}
  A large permanent impact $\alpha$, a large risk aversion $\phi$ and a small temporary impact $\kappa$ are the main components driving the net inventory of participants $E$ to zero.

\end{StylizedFact}

Another way to understand the optimal control $\nu$ is to look at its formulation not only as a function of $h_1$ and $h_2$, but at a function of $E$, $E'$ and $h_2$:
\begin{equation}
  \label{eq:ctrl:EEph2}
  \nu(t,q) = \frac{\partial_q v}{2\kappa} = \frac{1}{2\kappa } (h_1(t) - q \cdot h_2(t)) = E'(t) + \frac{E(t) - q}{2\kappa} \, h_2(t).
\end{equation}

Keeping in mind we formulated the problem from the viewpoint of a seller: $Q_0>0$ is the number of shares to be sold from $0$ to $T$, and $E_0>0$ means the net initial position of all participants is dominated by sellers. Since $E(t)$ decreases with time, $E'$ is rather negative.

The upper expression of the optimal control of our seller means the larger the remaining shares to sell $q(t)$, the faster to trade, proportionally to $h_2(t)$.
The influence of $A$ is clear: $h_2(T)=2A$ says the larger the terminal penalization, the faster to trade when $T$ is close, for a given remaining number of shares to sell.

Expression \eqref{eq:ctrl:EEph2} for the optimal control reads:
\begin{StylizedFact}[Influence of $E(t)$ and $E'(t)$ on the optimal control]
The optimal control is made of two parts: one ($-q h_2/(2\kappa)$) is proportional to the remaining quantity and independent of others' behavior; the other ($h_1=E' + E/(2\kappa)$) increases with the net inventory of other market participants and follows their trading flow.
Hence, in this framework:
\begin{itemize}
\item[(i)] it is optimal to ``follow the crowd'' (because of the $E'$ term)
\item[(ii)] but not too fast (since $E$ and $E'$ often have an opposite sign); especially when $t$ is close to $T$ (because of the $h_2$ term in factor of $E$).
\end{itemize}
\end{StylizedFact}

This pattern can be seen as a \emph{fire sales} pattern: the trader should follow participants while they trade in the same direction. This also means when the trader's inventory is opposite to market participants' net inventory, he can afford to slow down (because the price will be better for him soon).

\begin{StylizedFact}[Optimal trading speed with a very low inventory]
When a participant is close to a zero inventory (i.e. $q$ is close to zero) or for participant with low terminal constraints and low risk aversion, it can be rewarding to ``follow the crowd''.
The dominant term is then $h_1$: a sign change of $h_1$ implies a change of trading direction for a participant with a low inventory. 
Nevertheless once a participant followed $h_1$, his (no more neglectable) inventory multiplied by $h_2$ drives his trading speed.  
\end{StylizedFact}

Readers can have a look at the right panel Figure \ref{fig:E:small:h1} to (solid grey line) observe a sign change of $h_1$ (around $t=4$).

\paragraph{A specific case where $h_2$ is almost constant.}
When $A$ is very close to $\sqrt{\kappa\phi}$, expression (\ref{eq:h2:gen}) for $h_2$ shows $c_2$ will be very close to zero, and hence $h_2(t)\simeq 2\sqrt{\kappa\phi}\simeq 2A$ for any $t$.

\begin{figure}[!ht]
  \centering
  \includegraphics[width=\textwidth]{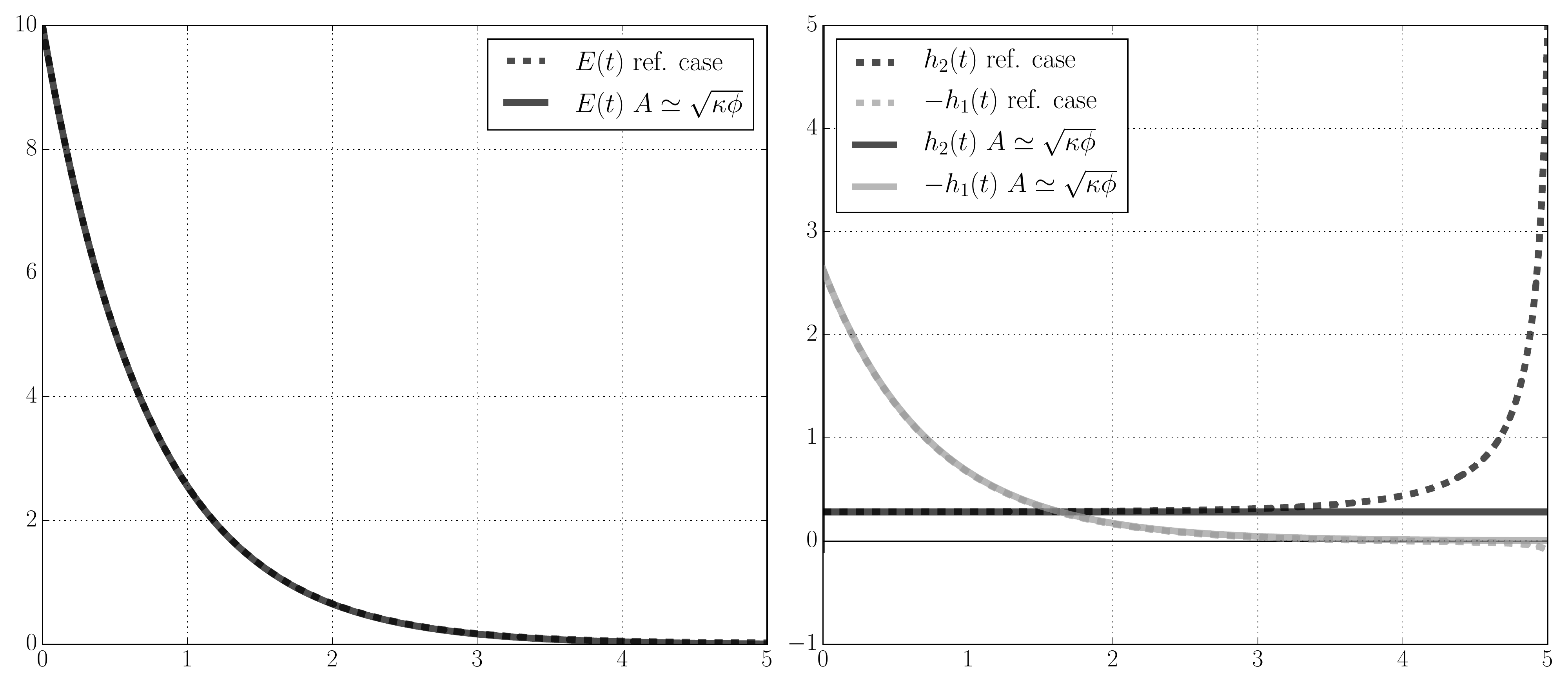}
  \caption{Comparison of the dynamics of $E$ (left) and $-h_1$ and $h_2$ (right) between the ``reference'' parameters of Figure \ref{fig:stdhEh} and when $\sqrt{\kappa\phi}\simeq A$: in such a case $h_2$ is almost constant but $E$ and $h_1$ are almost unchanged.}
  \label{fig:A:kp}
\end{figure}

Figure \ref{fig:A:kp} is an illustration of such a case: $A=\sqrt{\kappa\phi}=0.141$.
Since $E$ is not affected too much by this change in $A$ and remains close to zero when $t$ is large enough, the change of $h_2$ slope close to $T$ cannot affect significantly $h_1$.

\begin{StylizedFact}[Constant $h_2$\label{prop:cst:h2}]
  When $A=\sqrt{\kappa\phi}$, $h_2$ is constant (no more a function of $t$) equals to $2A$. 
  Hence the multiplier of $q$ (the remaining quantity) is a constant:
$$\nu(t,q) = E'(t) + \frac{E(t) - q}{\kappa} \,A.$$
\end{StylizedFact}

\paragraph{Specific configurations of $E$.}
When the considered trader has to sell while the net initial position of all the participants is to buy (i.e. $Q_0>0$ and $E_0<0$), the multiplier $h_2$ of the remaining quantity $q$ stays the same, but the constant term $h_1$ is turned into $-h_1$:

\begin{StylizedFact}[Alone against the crowd \label{prop:alone}]
 When the trader position does not have the same direction than the net inventory of all participants $E$:
 he has to trade slower, independently from his remaining inventory $q$.
\end{StylizedFact}

\begin{figure}[!h]
  \centering
  \includegraphics[width=\textwidth]{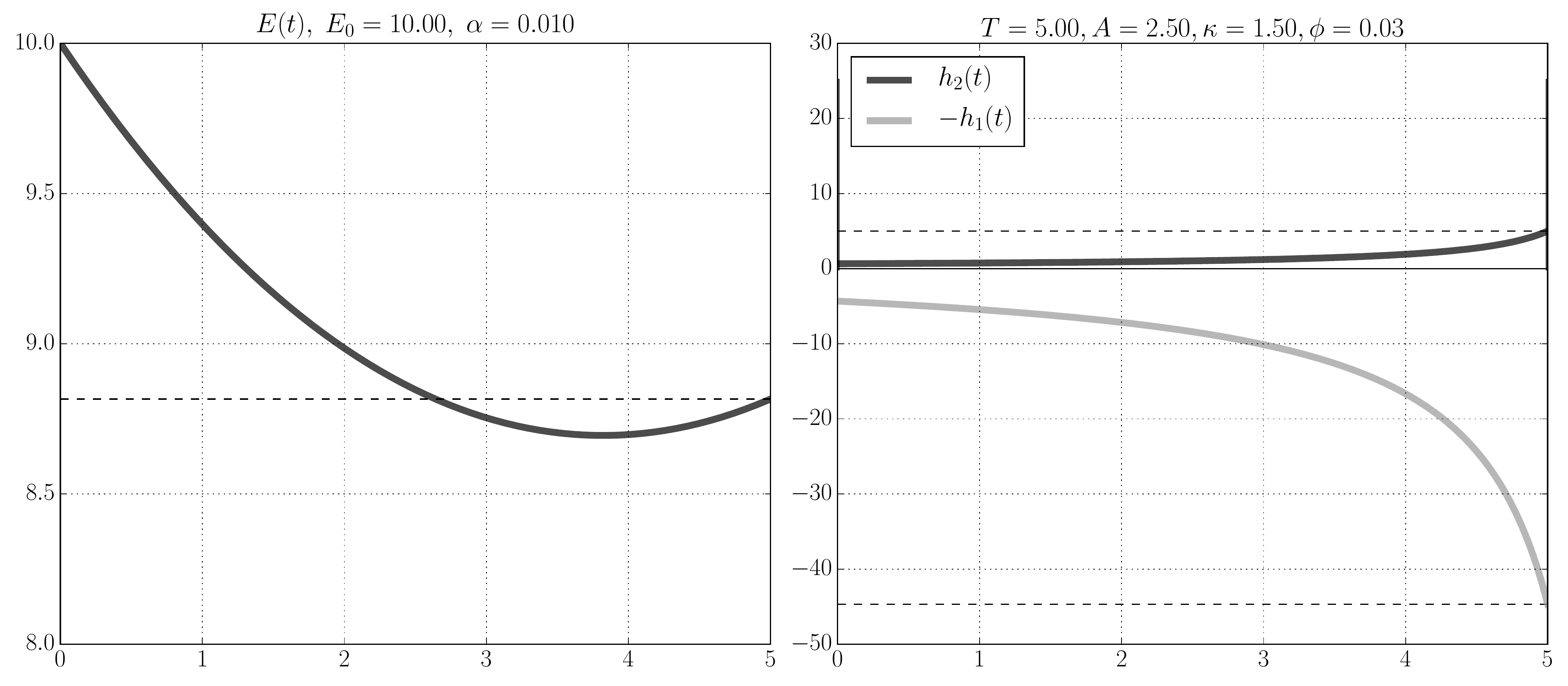}
  \caption{A specific case for which $E$ is not monotonous: $\alpha=0.01$, $\kappa=1.5$, $\phi=0.03$, $A=2.5$, $T=5$ and $E_0=10$. }
  \label{fig:E:inc}
\end{figure}

Moreover, the formulation of Stylized Fact \ref{prop:cst:h2} shows it is possible to change the monotony of $E(t)$ so that after a given $t$ is no more decreases:
$$E(t)'= E_0 a \underbrace{\left( r_+\exp\{r_ +t\}-r_-\exp\{r_- t\}\right) }_{\mbox{positive}}+ \underbrace{ E_0r_-\exp\{r_- t\}}_{\mbox{negative}}.
$$
For well chosen configuration of parameters the first term can be larger than the second term, for any $t$ greater than a critical $t^m$ such that $E'(t^m)=0$. For the configuration of Figure \ref{fig:E:inc}, with $\alpha=0.01$, $\kappa=1.5$, $\phi=0.03$, $A=2.5$, $T=5$ and $E_0=10$, we have $t^m\simeq 3.82$.

\answca{Going back to the meaning of this mean field game framework: it models dealing desks of asset managers receiving instructions from their portfolio managers to buy or sell large amounts of shares at the start of the day (or of the week). When $t^m<T$, it means that while the sum of initial instructions where to buy (respectively sell) this day, the ``mean field'' of dealing desks changed its mind: they did not strictly followed instructions. They bought (resp. sold) more than asked, and are now starting to sell back (resp. buy back) this temporary inventory to make profits by their own. Regulators could be interested by market parameters ($\alpha$, $\kappa$, $\phi$) allowing such configurations to appear.}

Figure \ref{fig:zE} shows configurations for which $t^m$ exists: small values of $\phi$ and $\alpha$ and large value of $\kappa$ are in favor of a small $t^m$. 
This means when the risk aversion and the permanent market impact coefficient are small while the temporary market impact is large, the slope of the net inventory of participants can have sign change.

\begin{figure}[!h]
  \centering
  \includegraphics[width=\textwidth]{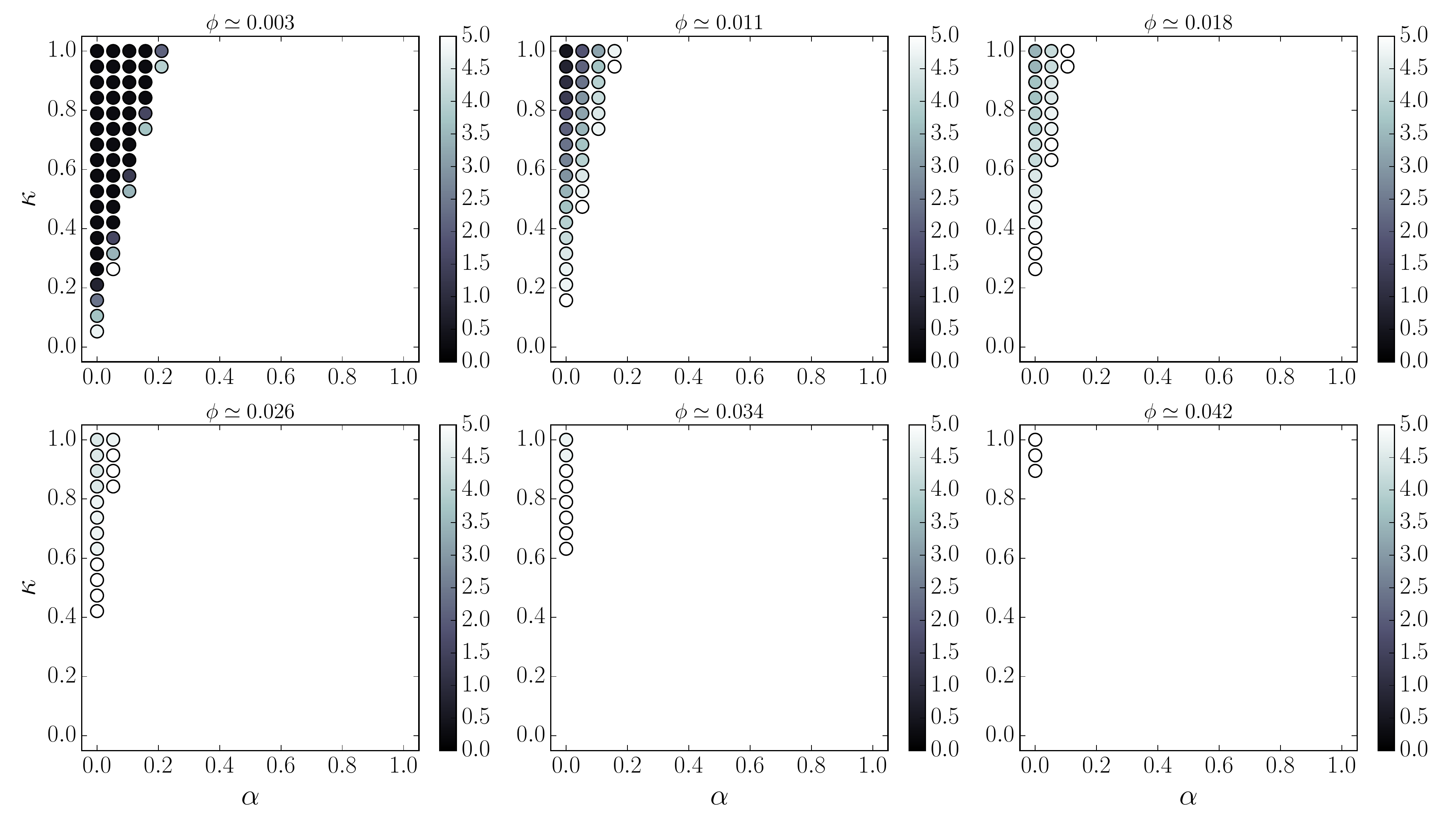}
  \caption{Numerical explorations of $t^m$ for different values of $\phi$ (very small $\phi$ at the top left to small $\phi$ at the bottom right) on the $\alpha\times\kappa$ plane, when $T=5$ and $A=2.5$. The color circles codes the value of $t^m$: small values (dark color) when $E$ changes its slope very early; large values (in light colors) when $E$ changes its slope close to $T$. }
  \label{fig:zE}
\end{figure}

\section{Trade crowding with heterogeneous preferences}
\label{sec:equilearn}

We now come back to our original model \eqref{eq:master} in which 
\answca{each} agent may be adverse to the risk 
a different way. In other words, the constants $A^a$ and $\phi^a$ may depend on $a$ 
\answca{but} are fixed 
\answca{for $0$ to $T$ (i.e. for the day or for the week)}. For simplicity, we will mostly work under the condition that $A^a= \sqrt{\phi^a \kappa}$, which allows to 
simplify the formulae. 

\subsection{Existence and uniqueness for the equilibrium model}
\label{sec:equilibrium}
As in the case of identical agents, we look for a solution to  system \eqref{eq:master} in which the map $v^a$ is quadratic: 
$$v^a(t,q) = h^a_0(t)+q\, h^a_1(t) - q^2 \, \frac{h^a_2(t)}{2},$$
and we find the relations:
\begin{equation}
 \label{eq:lqe1BIS}
 \begin{array}{rcl}
  0 &=& \displaystyle - 2\kappa (h^a_2)'- 4\kappa\phi^a+  (h^a_2)^2,\\
  \displaystyle 2\kappa \alpha\mu(t) &=&%
  \displaystyle  -2\kappa(h_1^a)' + h_1^ah_2^a,\\
  \displaystyle -(h_1^a)^2 &=& \displaystyle  4\kappa(h_0^a)' ,
\end{array}
\end{equation}
with terminal conditions 
\begin{equation}\label{eq:lqe1ICBIS}
h_0^a(T)=h_1^a(T)=0, \; h_2^a(T)= 2A^a.
\end{equation}
Let $m=m(t,dq,da)$ be the repartition at time $t$ of the wealth $q$ and of the parameter $a$. 
As before, the net trading flow $\mu$ at any time $t$ can be expressed as 
$$
  \mu(t) =  \int_{a,q} \frac{\partial_q v^a(t,q)}{2\kappa} \, m(t,da,dq) = \int_{a,q} \frac{1}{2\kappa}\left(h_1^a(t)-h_2^a(t)q\right)  \, m(t,da,dq).
$$
The measure $m$ solves the continuity equation 
$$
\partial_t m +\dive_q \left( m \frac{\partial_q v^a(t,q)}{2\kappa}\right)=0
$$
with initial condition $m_0=m_0(da,dq)$. For later use, we set 
$$
\bar m_0(da)= \int_q m_0(da,dq).
$$
As the agents do not change their parameter $a$ over the time, we always have 
$$
\int_q m(t,da,dq)= \bar m_0(da),
$$
so that we can disintegrate $m$ into 
$$
m(t,da,dq)= m^a(t,dq) \bar m_0(da),
$$
where $m^a(t,dq)$ is a probability measure in $q$ for $\bar m_0-$ almost any $a$. Let us set
$$
E^a(t)= \int_q q m^a(t,dq).
$$
Then by similar argument as in the case of identical agents  one has
\begin{equation}\label{eq:EaEa}
(E^a)'(t)= \frac{h_1^a(t)}{2\kappa}-\frac{h_2^a(t)}{2\kappa}E^a(t).
\end{equation}
With the notation $E^a$, we can rewrite $\mu$ as 
\begin{equation}\label{eq:mumu2}
\mu(t) = \frac{1}{2\kappa}\left(   \int_{a} h_1^a(t)\bar m_0(da) -\int_a h_2^a(t)E^a(t)\bar m_0(da)\right).
\end{equation}

\bigskip
From now on we assume {\it for simplicity of notation} that 
\begin{equation}\label{hyp.Aa=phikappa}
A^a = \sqrt{\phi^a\kappa}\qquad \forall a
\end{equation}
so that $h^a_2$ is constant in time with $h^a_2(t)= 2A^a$. We set 
$$
\theta^a = \frac{h_2^a}{2\kappa}.
$$
We solve the equation satisfied by $h_1^a$: as $h_1^a (T)=0$, we get
$$
h_1^a(t)= \alpha \int_t^T ds \exp\{ \theta^a(t-s)\} \mu(s).
$$
In the same way, we solve the equation for $E^a$ in \eqref{eq:EaEa} by taking into account the fact that $E^a(0)= E^a_0:=\int_q q \bar m_0(a,q)dq$ and the 
computation for $h_1^a$: we find
$$
E^a(t)= \exp\{ -\theta^a t\} E^a_0 +\frac{\alpha}{2\kappa}\int_0^t d\tau\int_\tau^T ds \exp\{\theta^a(2\tau-t-s)\} \mu(s).
$$
\answca{This can be read as the ``mean field'' associated to players with risk aversion parameters $(\phi^a,A^a)$.}

Putting these relations together \answca{and summing over all risk aversions (i.e. integrating over ${\bar m}_0(da)$)} we obtain by \eqref{eq:mumu2} that $\mu$ satisfies:
\begin{equation}\label{eq:fixedptmu}
\begin{array}{rl}
\ds \mu(t) \;  = &\ds -\int_a \theta^a\exp\{ -\theta^a t\} E_0^a\bar m_0(da) 
+\frac{\alpha}{2\kappa} \int_t^T ds\mu(s)\int_a \exp\{ \theta^a(t-s)\} \bar m_0(da)\\
& \ds 
-\frac{\alpha}{2\kappa} \int_0^t d\tau\int_\tau^T ds \mu(s) \int_a \theta^a \exp\{\theta^a(2\tau-t-s)\} \bar m_0(da).
\end{array}
\end{equation}

\begin{Proposition}\label{prop:eqHe} \answ{Assume that \eqref{hyp.Aa=phikappa} holds and that the $(A^a)$ are bounded.} Then there exists $\alpha_0>0$ such that, for $|\alpha|\leq \alpha_0$, there exists a unique solution to the fixed point relation \eqref{eq:fixedptmu}. 
\end{Proposition}

As a consequence, the MFG system has at least one solution obtained by plugging the fixed point $\mu$ into the relations \eqref{eq:lqe1BIS}.

\begin{proof} One just uses Banach fixed point Theorem on $\Phi_\alpha:C([0,T])\to C([0,T])$ which associates to any $\mu\in C([0,T])$ the map 
$$
\begin{array}{l}
\ds \Phi_\alpha(\mu)(t)\; := \;\ds -\int_a \theta^a\exp\{ \theta^a t\} E_0^a\bar m_0(da) +\frac{\alpha}{2\kappa} \int_t^T ds\mu(s)\int_a \exp\{ \theta^a(t-s)\} \bar m_0(da)\\
\qquad \qquad \qquad \qquad \ds-\frac{\alpha}{2\kappa} \int_0^t d\tau\int_\tau^T ds \mu(s) \int_a \theta^a \exp\{\theta^a(2\tau-t-s)\} \bar m_0(da).
\end{array}
$$
It is clear that $\Phi_\alpha$ is a contraction for $|\alpha|$ small enough{: \answca{i}ndeed, given $\mu,\mu'\in  C([0,T])$, we have 
$$
\begin{array}{rl}
\left|\Phi_\alpha(\mu)(t)- \Phi_\alpha(\mu')(t)\right| \; \leq & 
\ds \frac{|\alpha|}{2\kappa} \int_t^T ds\left|\mu(s)-\mu'(s)\right| \int_a \exp\{ \theta^a(t-s)\} \bar m_0(da)\\
& \ds+ \frac{|\alpha|}{2\kappa} \int_0^t d\tau\int_\tau^T ds \left| \mu(s) -\mu'(s)\right| \int_a \theta^a \exp\{\theta^a(2\tau-t-s)\} \bar m_0(da)\\
\leq & \; C|\alpha| \|\mu-\mu'\|_\infty,
 \end{array}
$$
for some constant $C$ independent of $\alpha$, because the $(\theta^a)$ are bounded. }
\end{proof}

\subsection{Learning other participants' net flows to (potentially) converge towards the MFG equilibrium}
\label{sec:learning}

The solution of the MFG system describing an equilibrium configuration, one may wonder how this configuration can be reached without the coordination of the agents. We present here a simple model to explain this phenomen\answca{on}
. For this we assume 
the game repeats the same $[0,T]$ intervals an infinite number of \emph{rounds}. 
The reader can think about $0$ to $T$ being a day (or a week), and hence each round will last a day (or a week). Round after round, market participants try to ``learn'' (i.e. to build an estimate of) the trading speed $\mu_t$ of equation \eqref{eq:price}.

It is close to what dealing desks of asset managers\footnote{Large asset managers, like Blackrock, Amundi, Fidelity or Allianz, delegate the implementation of their investment decisions to a dedicated (internal) team: their dealing desk. This team is in charge of trading to drive the real portfolios to their targets. They are implementing on a day to day basis what this paper is modelling.} are doing on financial markets: they try to estimate the buying or selling pressure exerted by other market participants to adjust their own behaviour.
That for, investment banks provide to their clients (asset managers) information about the recent state of the flows on exchanges, to help them to adjust their trading behaviours. For instance, JP Morgan's corporate and investment bank has a twenty pages recurrent publication titled ``\emph{Flows and Liquidity}'', with charts and tables describing different metrics like \emph{money flows}, turnovers, and other similar metrics by asset class (equity, bonds, options) and countries. Almost each investment bank has such a publication for its clients; Barclay's one is titled ``\emph{TIC Monthly Flows}''. As an example its August 2016 issue (9 pages) subtitle was ``\emph{Private foreign investors remain buyers}''.

Combining these expected flows with market impact models identified on their own flows (see \cite{citeulike:13497373}, \cite{citeulike:13266538} or \cite{citeulike:12825932} for academics papers on the topic), traders of dealing desks tune their optimal liquidation (or trading) schemes for the coming days. 
It is hence interesting to note this practice is very close to the framework we propose in this Subsection.

We also assume that, at the beginning of round $n$, each agent (generically indexed by $a$) has an estimate on $\mu^{a,n}$ on the crowd impact. Then he solves his corresponding optimal control problem: 
\begin{equation}\label{eq:HJlqn}
 \begin{array}{rcl}
  0 &=& \displaystyle - 2\kappa\frac{ (h^{a,n}_2)'}{2}- 4\kappa\phi^a+ (h^{a,n}_2)^2,\\
  \displaystyle 2\kappa\alpha\mu^{a,n}(t) &=&  \displaystyle - 2\kappa (h_1^{a,n})' + {h_1^{a,n}h_2^{a,n}},\\
  \displaystyle -{(h_1^{a,n})^2} &=& \displaystyle {4\kappa} (h_0^{a,n})' ,
\end{array}
\end{equation}
with terminal conditions 
$$
h_0^{a,n}(T)=h_1^{a,n}(T)=0, \; h_2^{a,n}(T)= 2A^a.
$$
For simplicity we assume \answca{once more} that 
$$
A^a = \sqrt{\phi^a\kappa}\qquad \forall a
$$
 so that $h^{a,n}_2$ is constant in time and independent of $n$: $h^a=2A^a$. We set as before
$$
\theta^a = \frac{h_2^a}{2\kappa}.
$$
We find as previously
$$
h_1^{a,n}(t)= \alpha \int_t^T ds \exp\{ \theta^a(t-s)\} \mu^{a,n}(s).
$$
and 
$$
E^{a,n}(t)= \exp\{ -\theta^a t\} E^a_0 +\frac{\alpha}{2\kappa}\int_0^t d\tau\int_\tau^T ds \exp\{\theta^a(2\tau-t-s)\} \mu^{a,n}(s).
$$
Then the net sum \answ{$m^{n+1}$} of the trading speed of all investors  at stage $n$ is given by \eqref{eq:mumu2} where $E^a$ and $h^a_1$ are replaced by $E^{a,n}$ and $h_1^{a,n}$:
\begin{equation}\label{eq:mn+1}
\begin{array}{rl}
\ds m^{n+1}(t) \;  := & \ds {
\frac{1}{2\kappa}\left(   \int_{a} h_1^{a,n}(t)\bar m_0(da) -\int_a h_2^aE^{a,n}(t)\bar m_0(da)\right)}\\
= & \ds -\int_a \theta^a\exp\{ -\theta^a t\} E_0^a\bar m_0(da) 
+\frac{\alpha}{2\kappa} \int_t^T ds\int_a \mu^{a,n}(s)\exp\{ \theta^a(t-s)\} \bar m_0(da)\\
& \ds 
-\frac{\alpha}{2\kappa} \int_0^t d\tau\int_\tau^T ds  \int_a\mu^{a,n}(s) \theta^a \exp\{\theta^a(2\tau-t-s)\} \bar m_0(da).
\end{array}
\end{equation}
We now model how the agents evaluate the crowd impact. We assume that each agent has his own way of evaluating the value of $m^{n+1}(t)$ and of incorporating this new value into his  estimate of the crowd impact. Namely, we suppose a relation of the form
$$
\mu^{a,n+1}(t):= (1-\pi^{a,n+1}) \mu^{a,n}(t)+\pi^{a,n+1}(m^{n+1}(t)+ \ep^{a,n+1}(t)),
$$
where $\pi^{a,n+1}\in (0,1)$ is the weight \answca{chosen by an agent $a$ at round $n+1$ to adjust his new estimate of $\mu$ with respect to his previous belief $\mu^{a,n}(t)$ and his new estimate $m^{n+1}(t)+ \ep^{a,n+1}(t)$} 
and $\ep^{a,n+1}$ is a small error term on the measured crowd impact: $\|\ep^{a,n+1}\|_\infty\leq \ep$. 

\begin{Proposition} 
\label{prop:learning} \answ{Under the assumption of Proposition \ref{prop:eqHe}}, 
let $\mu$ be the solution of the MFG crowd trading model. Assume that 
$$
\frac{1}{{C}n}\leq  \pi^{a,n}\leq \frac{C}{n},
$$
for some constants $C$.  Suppose also that  $\alpha$ is small enough: $|\alpha|\leq \alpha_1$ for some small $\alpha_1>0$ depending on $C$.
Then 
$$
\limsup \sup_a  \left\| \mu^{a,n}-\mu\right\|_\infty\leq C\ep
$$
for some constant $C$. 
\end{Proposition}

Let us note that, as the solution to the system \eqref{eq:HJlqn} depends in a continuous way of $\mu^{a,n}$, the optimal trading strategies of the agents are close to the one corresponding to the equilibrium configuration for $n$ sufficiently large. 

\begin{proof} In the proof the constant $C$ might differ from line to line, but does not depend on $n$ nor on $\ep$. We have
$$
\begin{array}{rl}
\ds \sup_a \|\mu^{a,n+1}-\mu\|_\infty \; \leq& \ds   \sup_a \|(1-\pi^{a,n+1}) \mu^{a,n}+ \pi^{a,n+1}(m^{n+1}+\ep^{a,n+1})-\mu\|_\infty\\
\leq & \ds    \sup_a \left( (1-\pi^{a,n+1}) \|\mu^{a,n}-\mu\|_\infty +\pi^{a,n+1}\|m^{n+1}-\mu\|_\infty+ \pi^{a,n+1}\|\ep^{a,n+1}\|_\infty\right)
\end{array}
$$
where, by \eqref{eq:mn+1}: 
$$
\|m^{n+1}-\mu\|_\infty\leq C\frac{\alpha}{\answca{2\kappa}} \sup_a \left\| \mu^{a,n}-\mu\right\|_\infty .
$$
So
$$
\sup_a \|\mu^{a,n+1}-\mu\|_\infty \; \leq \sup_a\left( (1-\pi^{a,n+1}) +C\frac{\alpha}{\answca{2\kappa}} \pi^{a,n+1}\right)\sup_a  \left\| \mu^{a,n}-\mu\right\|_\infty+
\sup_a\pi^{a,n+1}\ep.
$$
Thus, setting $\beta_n= \sup_a\left( (1-\pi^{a,n}) +C\frac{\alpha}{\answca{2\kappa}} \pi^{a,n}\right)$ and 
$\delta_n:= \sup_a\pi^{a,n}$, we have
$$
\sup_a  \left\| \mu^{a,n}-\mu\right\|_\infty\leq \sup_a  \left\| \mu^{a,0}-\mu\right\|_\infty\prod_{k=1}^{n}\beta_k
+ \ep\sum_{k=1}^{n} \delta_k \prod_{l=k+1}^{n}\beta_l .
$$
By our assumption on $\pi^{a,n}$, we can choose $\alpha$ small enough so that 
$$
 \beta_n\leq 1-\frac{1}{Cn}\qquad {\rm and}\qquad \delta_n \leq \frac{C}{n}.
$$
Then, for $1\leq k<n$, we have 
$$
\ln\left( \prod_{l=k}^{n}\beta_l\right)\leq \sum_{l=k}^{n}\ln(1-1/(Cl)) \leq -(1/C)  \sum_{l=k}^{n}\frac{1}{l}\leq -(1/C) \ln((n+1)/k).
$$
Hence $\displaystyle \prod_{l=k}^{n}\beta_l \leq  C\left(\frac{n}{k}\right)^{-1/C}$, which easily implies that
$$
\lim_n \prod_{k=1}^{n}\beta_k=0\qquad {\rm and}\qquad \sum_{k=1}^{n} \delta_k \prod_{l=k+1}^{n}\beta_l \leq C.
$$
The desired result follows.
\end{proof}

\section{A General Model for Mean Field Games of Controls} 
\label{sec:generic}

In this section we discuss a general existence result  for a Mean Field Game of Control (or, in the terminology of \cite{gomes2014existence}, an Extended Mean Field Game). As in our main application above, we aim at describing a system in which the {\it infinitely many small} agents control their own state and interact through a ``mean field of control". By ``small", we mean that the individual behavior of each agent has a negligible influence on the whole system. The requirement that they are ``infinitely many" corresponds to the fact that their initial configuration is distributed according to an absolutely continuous density on the state space. The ``mean field of control" consists in the  {\it joint distribution of the agents and their instantaneous control}: this is in contrats with the standard MFGs,  in which 
\answca{the mean field} is the distribution of the positions of the agents only. 

We denote by $X_t\in \R^d$  the individual state of a generic agent at time $t$ and by $\alpha_t$ his control. 
The state space is the finite dimensional space $\R^d$,  while the controls take their value in a metric space $(A,\delta_A)$. \answca{In this Section the} 
 distribution density of the pair $(X_t,\alpha_t)$  is denoted by $\mu_t$. It is a probability measure on $\R^d\times A$. The first marginal $m_t$ of $\mu_t$ is the distribution of the players at time $t$ (hence a probability measure on $\R^d$). In the MFG of control, dynamics and payoffs depend on $(\mu_t)$ (and not only on $(m_t)$ as in standard MFGs). 

We assume that the dynamics of a small agent is a controlled SDE of the form 
$$
\left\{\begin{array}{l}
dX_t =b(t,X_t, \alpha_t;\mu_t)dt + \sigma(t,X_t)dW_t\\
X_{t_0}=x_0
\end{array}\right.
$$
where $\alpha$ is the control and $W$ is a standard $D-$dimensional Brownian Motion (the Brownian Motions of the agents are independent).
\answ{Note that, for simplicity of notation, the heterogeneity of the agents (i.e., the parameter $a$ in the previous section) is encoded here in the state variable: it is a variable which is not affected by the dynamics. For this reason it is important to handle a degenerate diffusion term $\sigma$.}
The cost function is given by 
$$
J(t_0,x_0, \alpha; \mu)= \Esp \left[ \int_{t_0}^T L(t,X_t,\alpha_t;\mu_t)  \ dt+g(X_T,m_T)\right].
$$
It is known that, given $\mu$, the value function $u=u(t_0,x_0;\mu)$ of the agent, defined by 
$$
u(t_0,x_0)= \inf_\alpha J(t_0,x_0, \alpha; \mu),
$$
is a viscosity solution of the HJB equation 
\begin{equation}\label{eq:HJu}
\left\{\begin{array}{l}
-\partial_t u(t,x)-{\rm tr}(a(t,x)D^2u(t,x))+H(t,x,Du(t,x);\mu_t)= 0 \qquad {\rm in} \; (0,T)\times \R^d\\
u(T,x)= g(x)  \qquad {\rm in} \; \R^d
\end{array}\right.
\end{equation}
where $a=(a_{ij})=\sigma\sigma^T$ and 
\begin{equation}\label{def:H}
H(t,x,p;\nu)= \sup_{\alpha} \left\{ -p\cdot b(t,x,\alpha;\nu)- L(t,x,\alpha;\nu)\right\}.
\end{equation}
Moreover $b(t,x,\alpha^*(t,x);\mu_t):=-D_pH(t,x,Du(t,x);\mu_t)$ is (formally) the optimal drift for the agent at position $x$ and at time $t$. Thus, the population density $m=m_t(x)$ is expected to evolve according to the Kolmogorov equation
\begin{equation}\label{eq:Kolmo}
\left\{\begin{array}{l}
\partial_t m_t(x) -\sum_{i,j} \partial_{ij} (a_{ij}(t,x)m_t(x))-\dive \left( m_t(x)D_pH(t,x,Du(t,x);\mu_t)\right)=0 \\
\qquad \qquad \qquad \qquad \qquad \qquad \qquad \qquad \qquad \qquad \qquad \qquad \qquad \qquad  {\rm in} \;  (0,T)\times \R^d\\
m_0(x)= \bar m_0(x)  \qquad {\rm in} \; \R^d
\end{array}\right.
\end{equation}
Throughout this part we assume that the map $\alpha\to b(t,x,\alpha;\nu)$ is one-to-one with a smooth inverse and we denote by 
$\alpha^*=\alpha^*(t,x,p; \nu)$  the map which associates to any $p$ the unique control $\alpha^*\in A$ such that 
\begin{equation}\label{defPhi}
b(t,x,\alpha^*;\nu)= -D_pH(t,x,p;\nu).
\end{equation}
This means that, at each time $t$ and position $x$, the optimal control of a typically small player is $\alpha^*(t,x,Du(t,x); \mu_t)$. So, in an equilibrium configuration, the measure $\mu_t$ has to be the image of the measure $m_t$ by the map $x\to (x,\alpha^*(t,x,Du(t,x);\mu_t))$. This leads to the fixed-point relation:
\begin{equation}\label{def:mut}
\mu_t = \left( id, \alpha^*(t,\cdot,Du(t,\cdot);\mu_t)\right)\sharp m_t,
\end{equation}
\answ{where the right-hand side is the image of the measure $m_t$ by the map $x\to (x,  \alpha^*(t,x,Du(t,x);\mu_t))$.}
To summarize, the MFG of control takes the form: 
\begin{equation}\label{eq:MFGgen}
\left\{\begin{array}{ll}
(i) & -\partial_t u(t,x)-{\rm tr}(a(t,x)D^2u(t,x))+H(t,x,Du(t,x);\mu_t)= 0 \; {\rm in} \; (0,T)\times \R^d,\\
\qquad \\
(ii) & \partial_t m_t(x) -\sum_{i,j} \partial_{ij} (a_{ij}(t,x)m_t(x))-\dive \left( m_t(x)D_pH(t,x,Du(t,x);\mu_t)\right)=0 \\
 & \qquad \qquad \qquad \qquad \qquad \qquad \qquad \qquad \qquad \qquad \qquad \qquad    {\rm in} \; (0,T)\times \R^d,\\
(iii) & m_0(x)= \bar m_0(x), \; u(T,x)= g(x,m_T)  \qquad {\rm in} \; \R^d,\\
\qquad \\
(iv) & \mu_t = \left( id, \alpha^*(t,\cdot,Du(t,\cdot);\mu_t)\right)\sharp m_t \qquad {\rm in} \;[0,T].
\end{array}\right.
\end{equation}

The typical framework in which we expect to have a solution is the following: $u$ is continuous in $(t,x)$, Lipschitz continuous in $x$ (uniformly with respect to $t$) and satisfies equation \eqref{eq:HJu} in the viscosity sense; $m$ is in $L^\infty$ and satisfies \eqref{eq:Kolmo} in the sense of distribution. \\

In order to state the assumptions on the data, we need a few notation\answca{s}. Given a metric space $(E,\delta_E)$ we denote by ${\mathcal P}_1(E)$ the set of Borel probability measures $\nu$ on $E$ with a finite first order moment $M_1(\nu)$:
$$
M_1(\nu)= \int_E \delta_E(x_0, x)d\nu(x)<+\infty
$$
for some (and thus all) $x_0\in E$. The set  ${\mathcal P}_1(E)$ is endowed with the Monge-Kantorovitch distance: 
$$
{\bf d}_1(\nu_1,\nu_2)= \sup_{\phi}\int_E\phi(x)d(\nu_1-\nu_2)(x) \qquad \forall \nu_1,\nu_2\in  {\mathcal P}_1(E),
$$
where the supremum is taken over all 1-Lipschitz continuous maps $\phi:E\to \R$. 

We will prove the existence of a solution for \eqref{eq:MFGgen} under the following assumptions: 
\begin{enumerate}
\item The terminal cost $g:\R^d\times {\mathcal P}_1(\R^d)\to \R$ and the diffusion matrix $\sigma:[0,T]\times \R^d\to \R^{d\times D}$ are continuous and bounded, uniformly bounded in $C^2$ in the space variable,
\item\label{hypb} The drift has a separate form: $b(t,x,\alpha, \mu_t)= b_0(t,x,\mu_t)+ b_1(t,x,\alpha)$, 
\item\label{hypLL} The map $L:[0,T]\times \R^d\times A\times {\mathcal P}_1(\R^d\times A)\to \R$ satisfies the Lasry-Lions monotonicity condition: for any $\nu_1,\nu_2\in {\mathcal P}_1(\R^d\times A)$ with the same first marginal, 
$$
\int_{\R^d\times A} (L(t,x,\alpha;\nu_1)-L(t,x,\alpha;\nu_2))d(\nu_1-\nu_2)(x,\alpha)\geq 0, 
$$
\item\label{Condunik} For each $(t,x,p, \nu)\in [0,T]\times \R^d\times \R^d\times {\mathcal P}_1(\R^d\times A)$, there exists a unique maximum point $\alpha^*(t,x,p;\nu)$ in \eqref{defPhi} and  $\alpha^*:[0,T]\times \R^d\times \R^d\times {\mathcal P}_1(\R^d\times A)\to \R$ is continuous, with a linear growth: for any $L>0$, there exists $C_L>0$ such that 
$$
\delta_A(\alpha_0, \alpha^*(t,x,p;\nu))\leq C_L(|x|+1)\qquad \forall (t,x,p,\nu)\in [0,T]\times \R^d\times \R^d\times {\mathcal P}_1(\R^d\times A)\; {\rm with}\; |p|\leq L,
$$
(where $\alpha_0$ is a fixed element of $A$). 
\item The Hamiltonian $H:[0,T]\times \R^d\times \R^d\times {\mathcal P}_1(\R^d\times A)\to \R$ is continuous; $H$ is bounded in $C^2$ in $(x,p)$ uniformly with respect to $(t,\nu)$, and convex in $p$.
\item The initial measure $\bar m_0$ is a continuous probability density on $\R^d$ with a finite second order moment. 
\end{enumerate} 

The uniform bounds and uniform continuity assumptions are very strong requirement\answca{s}: for instance they are not satisfied in the linear-quadratic example studied before. These conditions can be relaxed in a more or less standard way; we choose not do so in order to keep the argument relatively simple. 

Following \cite{carmonadelaruebook}, assumptions \eqref{hypb} and \eqref{hypLL} ensure the uniqueness of the fixed-point in \eqref{def:mut}. For the sake of completeness, details are given in Lemma \ref{lem:PointFixe} below. 

\begin{Theorem}\label{theo:existencegen} Under the above assumptions, there exists at least one solution to the MFG system of controls \eqref{eq:MFGgen} for which  $(\mu_t)$ is continuous from $[0,T]$ to ${\mathcal P}_1(\R^d\times A)$.
\end{Theorem}

As almost always the case for this kind of results, the argument of proof consists in applying a fixed point argument of Schauder type and requires therefore compactness properties. The main difficulty is the control of the time regularity of the evolving measure $\mu$. This regularity is related with the time regularity of the optimal controls. In the first order case, it is not an issue because the optimal trajectories satisfy a Pontryagin maximum principle and thus are (at least) uniformly $C^1$ (see \cite{gomes2014existence,gomes2016extended}). In the second order setting, the Pontryagin maximum principle is not so simple to manipulate and a similar regularity for the controls would be much more heavy to express. We use instead the full power of semi-concavity of the value function combined with compactness arguments (see Lemma \ref{lem:unifcont} below). 
The main advantage of our approach is its robustness: for instance stability property of the solution is almost straightforward. 

The proof of Theorem \ref{theo:existencegen} requires several preliminary remarks. Let us start with the fixed-point relation \eqref{def:mut}. 

\begin{Lemma}\label{lem:PointFixe} Let $m\in {\mathcal P}_2(\R^d)$ with a bounded density and $p\in L^\infty(\R^d, \R^d)$. 
\begin{itemize}
\item (Existence and uniqueness.)
There exists a unique fixed point $\mu=F(p,m) \in {\mathcal P}_1(\R^d\times A)$ to the relation
\begin{equation}\label{fixedpointrel}
\mu= (id, \alpha^*(t,\cdot, p(\cdot); \mu))\sharp m.
\end{equation}
Moreover, there exists a constant $C_0$, depending only on $\|p\|_\infty$ and on the second order moment of $m$, such that 
$$
\int_{\R^d\times A} \left\{|x|^2+\delta_A(\alpha_0, \alpha)\right\}\; d \mu(x,\alpha) \leq C_0.
$$ 
\item (Stability.) Let $(m_n)$ be a family  of $ {\mathcal P}_1(\R^d)$, with a uniformly bounded density in $L^\infty$ and uniformly bounded second order moment, which converges in ${\mathcal P}_1(\R^d)$ to some $m$, $(p_n)$ be a uniformly bounded family in $L^\infty$ which converges a.e. to some $p$. Then $F(p_n,m_n)$ converges to $F(p,m)$ in ${\mathcal P}_1(\R^d\times A)$. 
\end{itemize}
\end{Lemma}

The uniqueness part is borrowed from \cite{carmonadelaruebook}. 

\begin{proof} Let $L:=\|p\|_\infty$. For $\mu\in {\mathcal P}_1(\R^d\times A)$, let us set $\Psi(\mu):=(id, \alpha^*(t,\cdot, p(\cdot); \mu))\sharp m$. Using the growth assumption on $\alpha^*$, we have
$$
\begin{array}{rl}
\ds \int_{\R^d\times A} |x|^2+\delta_A^2(\alpha_0,\alpha) \; d\Psi(\mu)(x,\alpha)\; = & \ds \int_{\R^d} |x|^2 + \delta_A^2(\alpha_0, \alpha^*(t,x, p(x); \mu)) \; m(x)dx \\
\leq & \ds  \int_{\R^d\times A} |x|^2+C_L^2(|x|+1)^2\;  m(x)dx =:C_0.
\end{array}
$$
This leads us to define ${\mathcal K}$ as the convex and compact subset of measures $\mu\in {\mathcal P}_1(\R^d\times A)$ with second order moment bounded by $C_0$. Let us check that the map $\Psi$ is continuous on ${\mathcal K}$. If $(\mu_n)$ converges to $\mu$ in ${\mathcal K}$, we have, for any map $\phi$, continuous and bounded on $\R^d\times A$:  
$$
\int_{\R^d\times A} \phi(x,\alpha) d\Psi(\mu_n)(x,\alpha) = 
\int_{\R^d} \phi(x,\alpha^*(t,x,p(x); \mu_n) m(x)dx.
$$
By continuity of $\alpha^*$, the term $\phi(x,\alpha^*(t,x,p(x); \mu_n))$ converges a.e. to $\phi(x,\alpha^*(t,x,p(x); \mu)$ and is bounded. The mesure $m$ having a bounded second order moment and being absolutely continuous, this implies the convergence of the integral: 
$$
\lim_n \int_{\R^d\times A} \phi(x,\alpha) d\Psi(\mu_n)(x,\alpha) = 
\int_{\R^d} \phi(x,\alpha^*(t,x,p(x); \mu)) m(x)dx
=
 \int_{\R^d\times A} \phi(x,\alpha) d\Psi(\mu)(x,\alpha).
$$
Thus the sequence $(\Psi(\mu_n))$ converges weakly to $\Psi(\mu)$ and, having a uniformly bounded second order moment, also converges for the ${\bf d}_1$ distance. The map $\Psi$ is continuous on the compact set ${\mathcal K}$ and therefore has a fixed point by Schauder fixed point Theorem. 

Let us now check the uniqueness. If there exists two fixed points $\mu_1$ and $\mu_2$  of \eqref{fixedpointrel}, then, by the monotonicity condition in Assumption \eqref{hypLL}, we have 
$$
\begin{array}{l}
\ds 0\leq \int_{\R^d\times A} \left\{L(t,x,\alpha;\mu_1)-L(t,x,\alpha;\mu_2)\right\}d(\mu_1-\mu_2)(x,\alpha) \\
\qquad \ds = \int_{\R^d} \left\{L(x,\alpha_1(x);\mu_1)-L(x,\alpha_1(x);\mu_2)-L(x,\alpha_2(x);\mu_1)+L(x,\alpha_2(x);\mu_2) \right\}m(x)dx
\end{array}
$$
where we have set $\alpha_i(x) :=\alpha^*(t,x,p(x);\mu_i)$ ($i=1,2$) and $L(x,\dots):= L(t,x,p(x),\dots)$ to simplify the expressions. 
So
$$
\begin{array}{l}
\ds 0\leq  \int_{\R^d} \left\{ b_1(t,x,\alpha_1(x))\cdot p(x)+L(x,\alpha_1(x);\mu_1)-b_1(t,x,\alpha_2(x))\cdot p(x)-L(x,\alpha_2(x);\mu_1) \right.\\
\qquad \qquad \ds \left. -b_1(t,x,\alpha_1(x))\cdot p(x)-L(x,\alpha_1(x);\mu_2)+b_1(t,x,\alpha_2(x)) \cdot p(x)+ L(x,\alpha_2(x);\mu_2) \right\} m(x)dx,
\end{array}
$$
 where, by assumption \eqref{Condunik}, $\alpha_1(x)$ is the unique maximum  point in the expression 
 $$
 -b_1(t,x,\alpha)\cdot p(x)-L(t,x,p(x), \alpha;\mu_1)
 $$ and $\alpha_2(x)$ the unique maximum point in the symmetric expression with $\mu_2$. This implies that $\alpha_1=\alpha_2$ $m-$a.e., and therefore, by the fixed point relation \eqref{fixedpointrel}, that $\mu_1=\mu_2$.  
 
 Finally we show the stability. In view of the previous discussion, we know that $(\mu_n:=F(p_n,m_n))$ has a uniformly bounded second order moment and thus converges, up to a subsequence, to some $\tilde \mu$. We just have to check that $\tilde \mu=F(p,m)$, i.e., $\tilde \mu$ satisfies the fixed-point relation. Let $\phi$ be a continuous and bounded map on $\R^d\times A$. Then 
 $$
 \int_{\R^d\times A} \phi(x,\alpha)d\mu_n(x,\alpha)= \int_{\R^d} \phi(x, \alpha^*(t,x,p_n(x); \mu_n))m_n(x)dx.
 $$
 As $(m_n)$ is bounded in $L^\infty$ and converges to $m$ in ${\mathcal P}_1(\R^d)$, it converges in $L^\infty-$weak-$*$ to $m$. The map $\phi$ being bounded and the $(m_n)$ having a uniformly bounded second order moment, we have therefore 
 $$
\lim_n  \int_{\R^d} \phi(x, \alpha^*(t,x,p(x); \mu))m_n(x)dx=  \int_{\R^d} \phi(x, \alpha^*(t,x,p(x); \mu))m(x)dx.
$$
On the other hand, by continuity of $\alpha^*$ and a.e. convergence of $(p_n)$,  $ \phi(x, \alpha^*(t,x,p_n(x); \mu_n))$ converges to 
$ \phi(x, \alpha^*(t,x,p(x); \mu))$ a.e. and thus in $L^1_{loc}$. As the $(m_n)$ are uniformly bounded in $L^\infty$ and have a uniformly bounded second order moment and as $\phi$ is bounded, this implies that 
$$
\lim_n  \int_{\R^d} \phi(x, \alpha^*(t,x,p_n(x); \mu_n))m_n(x)dx- \int_{\R^d} \phi(x, \alpha^*(t,x,p(x); \mu))m_n(x)dx=0.
$$
So we have  proved that 
$$
\lim_n  \int_{\R^d\times A} \phi(x,\alpha)d\mu_n(x,\alpha)= \int_{\R^d\times A} \phi(x,\alpha)d\mu(x,\alpha),
 $$
which implies that  $(\mu_n)$ converges to $\mu$ in ${\mathcal P}_1(\R^d\times A)$ because of the second order moment estimates on $(\mu_n)$. 
\end{proof} 

Next we address the existence of a solution to the Kolmogorov equation when the map $u$ is Lipschitz continuous and has some semi-concavity property in space. We will see in the proof of Theorem \ref{theo:existencegen} that this is typically the regularity for the solution of the Hamilton-Jacobi equation. 

\begin{Lemma}\label{lem.Kolmo} Assume that $u=u(t,x)$ is uniformly Lipschitz continuous in space with constant $M>0$ and semi-concave with respect to the space variable with constant $M$ and that $(\mu_t)$ is a time-measurable family in ${\mathcal P}_1(\R^d\times A)$. Then there exists at least one solution to the Kolmogorov equation \eqref{eq:Kolmo} which satisfies the bound
\begin{equation}\label{esti:bound}
\|m_t(\cdot)\|_\infty\leq \|\bar m_0\|_\infty\exp\{C_0t\},
\end{equation}
where $C_0$ is given by 
$$
C_0:=C_0(M)=\sup_{(t,x)}|D^2a(t,x)|+ \sup_{(t,x,p, \nu)} \left|D^2_{xp}H(t,x,p;\nu)\right|+ \sup_{(t,x,p,X,\nu)}{\rm Tr}\left(D^2_{pp}H(t,x,p;\nu)X\right)
$$
where the suprema are taken over the $(t,x,p,X,\nu)\in [0,T]\times \R^d\times \R^d\times {\mathcal S}^d\times {\mathcal P}_1(\R^d\times A)$ such that $\ |p|\leq M$ and $X\leq M\ Id$. Moreover, we have the second order moment estimate
\begin{equation}\label{est.:moment}
\int_{\R^d}|x|^2 m_t(x)dx \leq M_0
\end{equation}
and the continuity in time estimate
\begin{equation}\label{est.:conti}
\sup_{s,t} {\bf d}_1\left(m_{n,s}),m_{n,t}\right)\leq M_0|t-s|^{\frac12}\qquad \forall s,t\in [0,T], 
\end{equation}
where $M_0$ depends only on $\|a\|_\infty$, the second order moment of $\bar m_0$ and on $\sup_{t,x,p,\nu}|D_pH(t,x,p;\nu)|$, the supremum being taken over the $(t,x,p,\nu)\in [0,T]\times \R^d\times \R^d\times {\mathcal P}_1(\R^d\times A)$ such that $\ |p|\leq M$. \\ 

Assume now that $(u_n=u_n(t,x))$ is a family of continuous maps which are Lipschitz continuous and semi-concave in space with uniformly constant $M$ and converge locally uniformly to a map $u$; that $(\mu_n=\mu_{n,t})$ converges in $C^0([0,T], {\mathcal P}_1(\R^d\times A))$ to $\mu$. Let  $m_n$ be a solution to the Kolmogorov equation associated with $u_n$, $\mu_n$ which satisfies the $L^\infty$ bound \eqref{esti:bound}. Then $(m_n)$ converges, up to a subsequence, in $L^\infty-$weak*, to a solution $m$ of the Kolmogorov equation associated with $u$ and $\mu$. 
\end{Lemma}

\begin{proof} We first prove the existence of a solution. Let $(u_n)$ be  smooth approximations of $u_n$. Then equation 
$$
\left\{\begin{array}{l}
\partial_t m_t(x) -\sum_{i,j} \partial_{ij} (a_{n,ij}(t,x)m_t(x))-\dive \left( m_t(x)D_pH(t,x,Du_n(t,x);\mu_t)\right)=0\; {\rm in} \; (0,T)\times \R^d\\
m_0(x)= \bar m_0(x)  \qquad {\rm in} \; \R^d
\end{array}\right.
$$
has a unique classical solution $m_n$, which is the law of the process 
$$
\left\{\begin{array}{l}
dX_t =-D_pH(X_t,Du_n(t,X_t);\mu_t)dt + \sigma(t,X_t)dW_t\\
X_{0}={\bf x_0}
\end{array}\right.
$$
where ${\bf x_0}$ is a random variable with law $\bar m_0$  independent of $W$. 
By standard argument, we have therefore that 
$$
\sup_{t\in [0,T]} \int_{\R^d}|x|^2 m_{n,t}(dx) \leq \Esp \left[ \sup_{t\in [0,T]}|X_t|^2\right]\leq M_0,
$$
where $M_0$ depends only on $\|a\|_\infty$, the second order moment of $\bar m_0$ and on $\sup_{t,x,p,\nu}|D_pH(t,x,p;\nu)|$, the supremum being taken over the $(t,x,p,\nu)\in [0,T]\times \R^d\times \R^d\times {\mathcal P}_1(\R^d\times A)$ such that $\ |p|\leq M$.
Moreover, 
$$
 {\bf d}_1\left(m_n(s),m_n(t)\right)\leq \sup_{s,t} \Esp \left[ |X_t-X_s|\right]\leq M_0|t-s|^{\frac12}\qquad \forall s,t\in [0,T],
$$
changing $M_0$ if necessary. 
The key step is that $(m_n)$ is bounded in $L^\infty$. For this we rewrite the equation for $m_n$ as a linear equation in non-divergence form 
$$
\partial_t m_n -{\rm Tr}\left(a D^2m_n\right)+ b_n\cdot Dm_n - c_nm_n=0
$$
where $b_n$ is a some bounded vector field and $c_n$ is given by 
$$
c_n(t,x)= \sum_{i,j} \partial^2_{ij}a_{n,ij}(t,x)+ {\rm Tr}\left(D^2_{xp}H(t,x,Du_n(t,x);\mu_t)\right)+ {\rm Tr}\left(D^2_{pp}H(t,x,Du_n(t,x);\mu_t)D^2u_n(t,x)\right). 
$$
As $u$ is uniformly Lipschitz continuous and semi-concave with respect to the $x-$variable, we can assume that $(u_n)$ enjoys the same property, so that 
$$
\|Du_n\|_\infty \leq M, \qquad D^2u_n\leq M\ I_d.
$$
Then 
$$
|D^2a_n(t,x)|+\left|D^2_{xp}H(t,x,Du_n(t,x);\mu_t)\right|+ {\rm Tr}\left(D^2_{pp}H(t,x,Du_n(t,x);\mu_t)D^2u_n(t,x)\right)\leq C_0
$$
because $D^2_{pp}H\geq 0$ by convexity of  $H$. This proves that 
$$
c_n(t,x)\leq C_0 \qquad {\rm a.e.}
$$
By standard maximum principle, we have therefore that 
$$
\sup_{x} m_n(t,x)\leq \sup_{x} \bar m_0(x)\exp(C_0t)\qquad \forall t\geq0,
$$
which proves the uniform bound of $m_n$ in $L^\infty$. Thus $(m_n)$ converges weakly-* in $L^\infty$ to some map $m$ satisfying (in the sense of distribution)
$$
\left\{\begin{array}{l}
\partial_t m -\sum_{i,j} \partial_{ij} (a_{ij}m)-\dive \left( mD_pH(t,x,Du;\mu_t)\right)=0\; {\rm in} \; (0,T)\times \R^d\\
m_0(x)=\bar  m_0(x)  \qquad {\rm in} \; \R^d
\end{array}\right.
$$
This shows the existence of a solution. 
The proof of the stability goes along the same line, except that there is no need of regularization. 
\end{proof}

Next we need to show some uniform regularity in time of $Du$ where $u$ is a solution of the Hamilton-Jacobi equation. For this, let us note that the set
\begin{equation}\label{def:mathcalD}
{\mathcal D}:=\{ p\in L^\infty(\R^d), \; \exists v\in W^{1,\infty}(\R^d), \; p=Dv, \; \|v\|_\infty\leq M, \; \|Dv\|_\infty\leq M, \; D^2v\leq M\ I_d\}
\end{equation}
is sequentially compact for the a.e. convergence and therefore for the distance 
\begin{equation}\label{def:distonD}
d_{{\mathcal D}}(p_1,p_2)= \int_{\R^d} \frac{|p_1(x)-p_2(x)|}{(1+|x|)^{d+1}}dx\qquad \forall p_1,p_2\in {\mathcal D}. 
\end{equation}

\begin{Lemma}\label{lem:unifcont} There is a modulus $\omega$ such that, for any $(\mu_t)\in C^0([0,T], {\mathcal P}_1(\R^d\times A))$, the viscosity solution  $u$ to \eqref{eq:HJu} satisfies 
$$
d_{\mathcal D}(Du(t_1,\cdot), Du(t_2,\cdot)) \leq \omega (|t_1-t_2|)\qquad \forall t_1,t_2\in[0,T].
$$
\end{Lemma}

\begin{proof}  Suppose that the result does not hold. Then there exists $\epsilon>0$ such that, for any $n\in {\mathbb N} \backslash \{0\}$, there is $\mu_{n}\in C^0([0,T], {\mathcal P}_1(\R^d\times A))$ and times $t_n\in [0,T-h_n]$, $h_n\in [0,1/n]$ with 
$$
d_{\mathcal D}(Du_n(t_n,\cdot),Du_n(t_n+h_n,\cdot))\geq \ep, 
$$
where $u_n$ is the solution to \eqref{eq:HJu} associated with $\mu_{n}$. In view of our regularity assumption on $H$, the map $(x,p)\to H(t, x,p; \mu_{n,t})$ is locally uniformly bounded in $C^2$  independently of $t$ and $n$. So, there exists a time-measurable Hamiltonian $h=h(t,x,p)$, obtained as a weak limit of  $H(\cdot,\cdot,\cdot; \mu_{n,\cdot})$, such that, up to a subsequence,  $\int_0^t H(s, x,p; \mu_{n,s})ds$ converges locally uniformly in $x,p$ to $\int_0^t h(s,x,p)ds$. Note that $h$ inherits the regularity property of $H$ in $(x,p)$. Using the notion of $L^1-$viscosity solution and Barles' stability result of $L^1-$viscosity solutions for the weak (in time) convergence of the Hamiltonian \cite{barles2006new}, we can deduce that the $u_n$ converge locally uniformly to the $L^1-$viscosity solution $u$ of the Hamilton-Jacobi equation \eqref{eq:HJu}  associated with the Hamiltonian $h$. Without loss of generality, we can also assume that $(t_n)$ converges to some $t\in [0,T]$. 
Then, by semi-concavity,  $Du_n(t_n,\cdot)$ and $Du_n(t_n+h_n,\cdot)$ both converge a.e. to $Du(t,\cdot)$ because $u_n(t_n,\cdot)$ and $u_n(t_n+h_n,\cdot)$ both converge to $u(t,\cdot)$ locally uniformly. As the $Du_n$ are uniformly bounded in $L^\infty$, we conclude that 
$$
d_{\mathcal D}(Du_n(t_n,\cdot),Du_n(t_n+h_n,\cdot))\to 0, 
$$
and there is a contradiction. 
\end{proof}

\begin{proof}[Proof of Theorem \ref{theo:existencegen}] We proceed as usual by a fixed point argument. We first solve an approximate problem, in which we smoothen the Komogorov equation, and then we pass 
\answca{to} the limit. Let $\ep>0$ small, $\xi^\ep=\ep^{-(d+1)}\xi((t,x)/\ep)$ be a standard smooth mollifier. 

 To any $\mu=(\mu_t)$ belonging to $C^0([0,T],{\mathcal P}_1(\R^d\times A))$ we associate the unique viscosity solution $u$ to \eqref{eq:HJu}. Note that, with our assumption on $H$ and $G$, $u$ is uniformly bounded, uniformly continuous in $(t,x)$ (uniformly with respect to $(\mu_t)$), uniformly Lipschitz continuous and semi-concave in $x$ (uniformly with respect to $t$ and $(\mu_t)$). We denote by $M$  the uniform bound on the $L^\infty-$norm, the Lipschitz constant and semi-concavity constant: 
\begin{equation}\label{defMdefM}
\|u\|_\infty\leq M, \qquad \|Du\|_\infty\leq M, \qquad D^2u\leq M\ I_d.
\end{equation}
Then we consider $(m_t)$ to be the unique solution to the (smoothened) Kolmogorov equation  
\begin{equation}\label{eq.KolmoEps}
\left\{\begin{array}{l}
\partial_t m_t(x) -\sum_{i,j} \partial_{ij} (a_{ij}(t,x)m_t(x))-\dive \left( m_t(x)D_pH(t,x,Du^\ep(t,x);\mu_t)\right)=0 \\
\qquad \qquad \qquad \qquad \qquad \qquad \qquad \qquad \qquad \qquad \qquad \qquad \qquad \qquad  {\rm in} \;  (0,T)\times \R^d\\
m_0(x)= \bar m_0(x)  \qquad {\rm in} \; \R^d
\end{array}\right.
\end{equation}
where $u^\ep= u\star \xi^\ep$. Following Lemma \ref{lem.Kolmo}, the solution $m$---which is unique thanks to the space regularity of the drift---satisfies the bounds \eqref{esti:bound}, \eqref{est.:moment} and \eqref{est.:conti} (which are independent of $\mu$ and $\ep$).  Finally, we set $\tilde \mu_t= F(m_t, Du(t,\cdot))$, where $F$ is defined in Lemma \ref{lem:PointFixe}.  From Lemma \ref{lem:PointFixe}, we know that there exists $C_0>0$ (still independent of $\mu$ and $\ep$) such that 
$$
\sup_{t\in [0,T]} \int_{\R^d\times A} |x|^2+\delta_A(\alpha_0, \alpha)\; d\tilde \mu_t(x,\alpha) \leq C_0. 
$$

Our aim is to check that the map $\Psi^\ep:(\mu_t)\to (\tilde \mu_t)$ has a fixed point. Let us first prove that $(\tilde \mu_t)$ is uniformly continuous in $t$ with a modulus $\omega$ independent of $(\mu_t)$ and $\ep$. Recall that the set ${\mathcal D}$ defined by \eqref{def:mathcalD} is compact. 
Moreover, the subset ${\mathcal M}$   of measures  in ${\mathcal P}_1(\R^d)$ with a second moment bounded by a constant $C_0$ is also compact. By Lemma \ref{lem:PointFixe}, the map $F$ is continuous on the compact set ${\mathcal D}\times {\mathcal M}$, and thus uniformly continuous. On the other hand, Lemma \ref{lem:unifcont} states that the map $t\to Du(t,\cdot)$ is continuous from $[0,T]$ to ${\mathcal D}$, with a modulus independent of $(\mu_t)$. 
 As $(m_t)$ is also uniformly continuous in time (recall \eqref{est.:conti}), we deduce that $(\tilde \mu_t:= F(m_t, Du(t,\cdot))$ is uniformly continuous with a modulus $\omega$ independent of $(\mu_t)$ and $\ep$. 

Let ${\mathcal K}$ be the set of $\mu\in C^0([0,T], {\mathcal P}(\R^d\times A))$ with a second order moment bounded by $C_0$ and
a modulus of time-continuity $\omega$. Note that ${\mathcal K}$ is convex and compact and that  $\Psi^\ep$  is defined from ${\mathcal K}$ to ${\mathcal K}$. 

Next we show that the map $\mu\to \Psi^\ep(\mu)$ is continuous on ${\mathcal K}$. Let $(\mu_{n})$ converge to $\mu$  in ${\mathcal K}$. From the stability of viscosity solution, the solution $u_n$  of \eqref{eq:HJu} associated to $\mu_{n}$ converge locally uniformly to the solution $u$ associated with $\mu$. Moreover, recalling the estimates \eqref{defMdefM}, the $(Du_n)$ are uniformly bounded and converge a.e. to $Du$ by semi-concavity. Let $m_{n}$ be the solution to \eqref{eq.KolmoEps} associated with $\mu_n$ and $u_n$. By the stability part of Lemma \ref{lem.Kolmo}, any subsequence of the compact family $(m_n)$ converges uniformly to a solution  \eqref{eq.KolmoEps}. This solution $m$ being unique,  the full sequence $(m_n)$ converges to $m$. By continuity of the map $F$, we then have, for any $t\in [0,T]$,  
$$
\Psi^\ep(\mu_n)(t)=F(m_t, Du(t,\cdot))= \lim_n F(m_{n,t}, Du_n(t,\cdot))= \lim_n \Psi^\ep(\mu_n)(t).
$$ 
Since the  $(m_{n})$ and $(Du_n)$ are uniformly continuous in time, the convergence of $\Psi^\ep(\mu_n)$ to $\Psi^\ep(\mu)$ actually holds in ${\mathcal K}$. 

So, by Schauder fixed point Theorem, $\Psi^\ep$ has a fixed point $\mu^\ep$. It remains to let $\ep\to0$. Up to a subsequence, labelled in the same way, $(\mu^\ep)$ converges to some $\mu\in {\mathcal K}$. As above, the solution $u^\ep$ of \eqref{eq:HJu} associated with $\mu^\ep$ converges to the solution $u$ associated with $\mu$ and $Du^\ep$ converges to $Du$ a.e.. The solution $m^\ep$ of \eqref{eq.KolmoEps} (with $\mu^\ep$ and $u^\ep$) satisfies the estimates \eqref{esti:bound}, \eqref{est.:moment} and \eqref{est.:conti}. Thus the stability result in Lemma \ref{lem.Kolmo} implies that a subsequence, still denoted $(m^\ep)$, converges uniformly to a solution $m$ of the unperturbed Kolmogorov equation \eqref{eq:Kolmo}. Recall that $\mu^\ep_t= F(m^\ep_t, u^\ep(t,\cdot))$ for any $t\in [0,T]$. We can pass to the limit in this expression to get  $\mu_t= F(m_t, u(t,\cdot))$ for any $t\in [0,T]$. Then the triple $(u,m,\mu)$ satisfies system \eqref{eq:MFGgen}. 
\end{proof}

\section{Conclusion}

In this paper, we proposed a model for optimal execution or optimal trading in which, instead of having as usual a large trader in front of a neutral ``background noise'',  the trader is surrounded by a continuum of other market participants. The cost functions of these participants have the same shape, without necessarily sharing the same values of parameters.
Each player of this mean field game (MFG) is similar in the following sense: it suffers from temporary market impact, impacts permanently the price, and fears uncertainty.

The stake of such a framework is to provide robust trading strategies to asset managers, and to shed light on the price formation process for regulators.

Our framework is not a traditional MFG, but falls into the class of \emph{extended mean filed games}, in which participants interact via their controls. We provide a generic framework to address it for a vast class of cost functions, beyond the scope of our model. 

Thanks to it, we solved our model and provide insights on the influence of its parameters (temporary and permanent market impact coefficients, terminal penalization, risk aversion and duration of the game) on the obtained results. We provide the solution in a closed form and formulate a series of ``stylized facts'' (Stylized Fact \ref{prop:smallalpha} to Stylized Fact \ref{prop:alone}) describing our results. For instance we unveil three components of the optimal control: two coming from the mean field $E(t)$ and its derivative $E'(t)$ (summarized in a function $h_1(t)$), and the third one proportional to the remaning quantity to trade $q(t)$ via an increasing function $h_2(t)$.
We show also how to slow down trading when the net inventory of the participants is of the opposite sign (i.e. the ``market'' is buying while the trader is selling, or the reverse): $h_2(t)$ is unchanged but $h_1(t)$ is changed in its opposite.

\answca{To conclude on this, we provide numerical illustrations showing market participants could end up not following their initial instructions, for some configurations of the market structure. This could help regulators to smooth such behaviours if needed.}

In a second stage, we address the case of heterogenous preferences (i.e. when each agent has his own risk aversion parameter). We show the existence of an unique solution but do no more have a closed form formulation.
Last but not least we study a more realistic case in which participants do not know instantaneously the optimal strategies of others, but have \emph{to learn them}. We list in Proposition \ref{prop:learning} conditions needed so that the learnt strategy is the optimal one.

\appendix
\section{Proof of Proposition \ref{prop:E}}
\label{sec:app:prop:E}

\begin{proof}[Proof of Proposition \ref{prop:E}] 

The discriminant of the second order equation is
$$
\Delta = \alpha^2+16\kappa\phi
$$
and the roots are
$$
r_\pm = -\frac{\alpha}{4\kappa} \pm\frac{1}{\kappa} \sqrt{\kappa \phi+\frac{\alpha^2}{16}}.
$$
Hence 
$$
E(t)= E_0a \left( \exp\{r_+t\}-\exp\{r_-t\}\right)+ E_0\exp\{r_-t\}
$$
where $a\in \R$  determined by the condition 
$$
 \kappa  E'(T) + AE(T)=0
 $$
 and thus has to solve the relation 
$$
\begin{array}{l}
\ds \kappa E_0\left[a \left(r_+ \exp\{r_+T\}-r_-\exp\{r_-T\}\right)+ r_-\exp\{r_-T\}\right] \\
\qquad \qquad \qquad \ds +A E_0\left[a \left( \exp\{r_+T\}-\exp\{r_-T\}\right)+ \exp\{r_-T\}\right]=0.
\end{array}
$$
There is a unique solution if $E_0\neq 0$ and 
\begin{equation}\label{eq:condcond}
\kappa \left(r_+ \exp\{r_+T\}-r_-\exp\{r_-T\}\right)+A \left( \exp\{r_+T\}-\exp\{r_-T\}\right)\neq 0.
\end{equation}
Writing $\ds r^\pm = -\frac{\alpha}{4\kappa} \pm \theta$ where $\theta:=\frac{1}{\kappa} \sqrt{\kappa\phi+\frac{\alpha^2}{16}}$, condition \eqref{eq:condcond} is equivalent to 
$$
\left[(-\frac{\alpha}{4}+\kappa\theta) \exp\{\theta T\}-(-\frac{\alpha}{4}-\kappa\theta)\exp\{-\theta T\}\right]+A \left[ \exp\{\theta T\}-\exp\{-\theta T\}\right]\neq 0,
$$
which leads to the condition
$$
-\frac{\alpha}{2}{\rm sh}\{\theta T\} +2\kappa\theta{\rm ch}\{\theta T\}+2A {\rm sh}\{\theta T\}\neq 0.
$$
As ${\rm ch}\{\theta T\}> {\rm sh}\{\theta T\}$, one has 
$$
\begin{array}{l}
\ds -\frac{\alpha}{2}{\rm sh}\{\theta T\} +2\kappa\theta{\rm ch}\{\theta T\}+2A {\rm sh}\{\theta T\}\; > \; \ds  {\rm sh}\{\theta T\} \left( -\frac{\alpha}{2}+2\kappa\theta+2A \right)\\
\qquad\qquad  =  \ds {\rm sh}\{\theta T\} \left( -\frac{\alpha}{2}+2\left(\kappa\phi+\frac{\alpha^2}{16}\right)^{1/2}+2A \right)\;  \geq \;2A {\rm sh}\{\theta T\}\; >\; 0.
\end{array}
$$
So condition \eqref{eq:condcond} is always fulfilled and
$$
\begin{array}{rl}
\ds a\; = & \ds - \frac{\kappa r_- \exp\{r_- T\} +A \exp\{ r_- T\}}{\kappa \left(r_+ \exp\{r_+T\}-r_-\exp\{r_-T\}\right)+A \left( \exp\{r_+T\}-\exp\{r_-T\}\right)} \\
= & \ds \ds - \frac{(-\alpha/4-\kappa\theta+A) \exp\{-\theta T\} }{-\frac{\alpha}{2}{\rm sh}\{\theta T\} +2\kappa\theta{\rm ch}\{\theta T\}+2A {\rm sh}\{\theta T\}}.
\end{array}
$$

To compute $h_2$, we note that it solves the following backward ordinary differential equation \eqref{eq:dh2}:
$$\left\{\begin{array}{rcl}
    0 &=& 2\kappa \cdot h'_2(t) + 4 \kappa \cdot \phi - (h_2(t))^2\\
    h_2(T) &=& 2A
\end{array}\right.$$
It is easy to check the solution is given by \eqref{eq:h2:gen}, where 
where $r=2\sqrt{\phi/\kappa}$ and $c_2$ solves the terminal condition. Hence
$$c_2=\frac{1 - A/\sqrt{\kappa \phi}}{1 + A/\sqrt{\kappa\phi}} \cdot e^{-rT}.$$
\end{proof}

\if t\vMAFE %

\else

\fi

\end{document}